\documentclass[journal,10 pt,oneside,web]{ieeecolor} 

\usepackage{graphicx}
\usepackage{textcomp}
\usepackage{epsfig} 
\usepackage{epsfig,color,amsmath,cite}

\usepackage{wrapfig}
\usepackage{bm}
\usepackage{epstopdf}
\usepackage{amssymb}
\usepackage{url}

\usepackage{multirow}
\usepackage{hhline}
\usepackage{booktabs}
\usepackage[linesnumbered,boxed,commentsnumbered,ruled,vlined,longend]{algorithm2e}
\usepackage{comment}

\DeclareMathOperator{\trace}{tr}
\DeclareMathOperator{\rank}{rank}

\DeclareMathOperator*{\argmax}{arg\ max}

\makeatother
\DeclareMathAlphabet\mathbfcal{OMS}{cmsy}{b}{n}



\usepackage{stackengine}

\newcommand{\mat}[1]{\boldsymbol{#1}}

\providecommand{\eye}{\mat{I}}
\providecommand{\mA}{\ensuremath{\mat{A}}}

\providecommand{\mC}{\ensuremath{\mat{C}}}

\newcommand{\st}{{\rm s.t.}}

\newcommand{\m}{\boldsymbol}
\allowdisplaybreaks[4]
\usepackage[colorlinks = true,
linkcolor = blue,
urlcolor  = black,
citecolor = blue,
anchorcolor = blue]{hyperref}

\newcommand{\mb}[1]{\mathbf{#1}}
\newcommand{\mc}[1]{\mathcal{#1}}
\newcommand{\mbb}[1]{\mathbb{#1}}
\newcommand{\mr}[1]{\mathrm{#1}}
\usepackage[framemethod=TikZ]{mdframed}
\mdfdefinestyle{MyFrame}{%
	linecolor=black,
	outerlinewidth=1.25pt,
	roundcorner=1.25pt,
	innerrightmargin=5pt,
	innerleftmargin=5pt,}
	
\usepackage[noabbrev]{cleveref}

\usepackage{mathtools}

\DeclarePairedDelimiter\abs{\lvert}{\rvert}%
\DeclarePairedDelimiter\norm{\lVert}{\rVert}%

\makeatletter
\let\oldabs\abs
\def\abs{\@ifstar{\oldabs}{\oldabs*}}
\let\oldnorm\norm
\def\norm{\@ifstar{\oldnorm}{\oldnorm*}}
\makeatother

\usepackage[english]{babel}
\usepackage[utf8]{inputenc}
\usepackage[super]{nth}

\usepackage{graphicx}
\usepackage{float}
\usepackage[caption = false]{subfig}

\usepackage{array}
\usepackage{threeparttable}

\usepackage[english]{babel}
\usepackage[utf8]{inputenc}
\usepackage[super]{nth}

\usepackage{pifont}

\usepackage{mathtools}

\usepackage{mleftright}
\usepackage{xparse}

\NewDocumentCommand{\evaluat}{sO{\big}mm}{%
	\IfBooleanTF{#1}
	{\mleft. #3 \mright|_{#4}}
	{#3#2|_{#4}}%
}
\let\labelindent\relax
\usepackage{enumitem}
\setlist[itemize]{leftmargin=*}

\usepackage{lettrine} 

\usepackage{balance} 

\usepackage{colortbl}

\usepackage{pgfplots}
\pgfplotsset{compat=1.18}

\usepackage{tikz}
\usetikzlibrary{shapes.geometric, arrows.meta, decorations.pathmorphing, decorations.markings, calc, 3d}
\usetikzlibrary{positioning, fit, backgrounds}
\usetikzlibrary{arrows.meta}  
\usetikzlibrary{positioning}
\usetikzlibrary{fit}		
\usetikzlibrary{decorations.pathreplacing, calc}
\usetikzlibrary{tikzmark, calc} 
\tikzset{
	tangent/.style={decoration={
			markings,mark=at position #1 with {
				\coordinate (ta) at (0,0);
				\coordinate (tb) at (0.1,0);
			}
		},postaction=decorate},
	tangent/.default=0.5
}

\definecolor{Vgray}{RGB}{247,247,247}
\definecolor{Vedge}{RGB}{200,200,200}

\definecolor{Sub1purple}{RGB}{200,200,255}
\definecolor{Sub1edge}{RGB}{205,215,238}

\definecolor{Sub2yellow}{RGB}{251,234,206}
\definecolor{Sub2edge}{RGB}{251,234,206} 

\definecolor{Sub3green}{RGB}{215,238,210}
\definecolor{Sub3edge}{RGB}{215,232,210}

\definecolor{Sub1color}{rgb}{0.85, 0.9, 1}
\definecolor{Sub2color}{rgb}{1, 0.9, 0.7}
\definecolor{Sub3color}{rgb}{0.8, 1, 0.8}
\definecolor{Sub4color}{RGB}{200,200,255}

\definecolor{Sub4edge}{RGB}{205,215,238}

\definecolor{Sub5color}{RGB}{255,210,210}
\definecolor{Sub5edge}{RGB}{240,180,180}

\definecolor{Sub5edge}{RGB}{140,120,80}
\definecolor{Sub5color}{RGB}{255,245,225}

\definecolor{Sub6edge}{RGB}{100,130,150}
\definecolor{Sub6color}{RGB}{220,240,255}

\definecolor{Sub7edge}{RGB}{130,100,150}
\definecolor{Sub7color}{RGB}{245,230,255}

\definecolor{Sub8edge}{RGB}{150,130,100}
\definecolor{Sub8color}{RGB}{255,250,230}

\definecolor{Sub9edge}{RGB}{90,140,140}
\definecolor{Sub9color}{RGB}{220,250,250}

\definecolor{Sub10edge}{RGB}{110,110,110}
\definecolor{Sub10color}{RGB}{240,240,240}

\tikzset{
	highlightP1/.style={
		fill=Sub2yellow, rounded corners=2pt, inner sep=2pt, anchor=base},
	highlightP2/.style={
		fill=Sub3green, rounded corners=2pt, inner sep=2pt, anchor=base}
}

\definecolor{LightBlue}{rgb}{0.7, 0.85, 1}
\definecolor{LightOrange}{rgb}{1, 0.8, 0.6}
\definecolor{LightGreen}{rgb}{0.843, 0.933, 0.824}

\tikzstyle{box} = [rectangle, draw, thick, rounded corners, 
text centered, font=\footnotesize, minimum height=1cm, text width=4.5cm]
\tikzstyle{colorbox1} = [rectangle, draw=black, rounded corners, fill=Sub1purple, 
text centered, font=\footnotesize, minimum height=1cm, text width=4.5cm]
\tikzstyle{colorbox2} = [rectangle, draw=black, rounded corners, fill=Sub2yellow, 
text centered, font=\footnotesize, minimum height=1cm, text width=4.5cm]
\tikzstyle{colorbox3} = [rectangle, draw=black, rounded corners, fill=Sub3green, 
text centered, font=\footnotesize, minimum height=1cm, text width=4.5cm]
\tikzstyle{colorbox4} = [rectangle, draw=black, rounded corners, fill={rgb,255:red,247; green,247; blue,247},text centered, font=\scriptsize, minimum height=0.5cm, text width=3cm]

\tikzstyle{arrow} = [thick,->,>=stealth]
\tikzstyle{dashedarrow} = [thick, dashed, ->, >=stealth]
\usepackage{listings}
\usepackage{circuitikz}

\usepackage{makecell}

\usepackage{flowchart}

\newcommand{\introstart}[2]{\noindent \lettrine[lines=2]{#1}{#2}}

\definecolor{LCSS}{HTML}{004498}

\newcommand{\parstartc}[1]{\noindent \textbf{\textcolor{LCSS}{#1.}}\;}
\newcommand{\parquestc}[1]{\noindent \textbf{\textcolor{LCSS}{#1?}}\;}

\newcommand{\R}{\mathbb{R}}
\newcommand{\Rn}[1]{\mathbb{R}^{#1}}

\newcommand{\logdet}[1]{\mr{log} \mr{det}(#1)}

\newcommand{\titlesc}[1]{\title{\Large \vspace{0.7cm} \LARGE \centering {\textsc{{#1}}}}}

\renewenvironment{proof}{%
	\par\vspace{0.3em}%
	\noindent\textbf{\textit{\textcolor{LCSS}{Proof.}}}\enspace%
}{%
	\hfill$\blacksquare$%
	\par\vspace{0.3em}%
}
\newtheorem{theorem}{\textbf{\textcolor{LCSS}{Theorem}}}
\newtheorem{mylem}{\textbf{\textcolor{LCSS}{Lemma}}}
\newtheorem{mydef}{\textbf{\textcolor{LCSS}{Definition}}}

\newtheorem{myrem}{\textbf{\textcolor{LCSS}{Remark}}}

\newtheorem{mycor}{\textbf{\textcolor{LCSS}{Corollary}}}
\newtheorem{myprs}{\textbf{\textcolor{LCSS}{Proposition}}}

\newtheorem{problem}{\textbf{\textcolor{LCSS}{Problem}}}

\makeatletter
\def\@begintheorem#1#2{%
	\par\noindent\textbf{\textcolor{LCSS}{\textit{#1 #2}}}%
	\@ifnextchar[{\@withinfo}{\textbf{\textcolor{LCSS}{.}}\enspace\ignorespaces}}
	
	\def\@withinfo[#1]{
\textbf{\textcolor{LCSS}{\textit{ (#1)}}}\textcolor{LCSS}{.} \ignorespaces}
\makeatother

\usepackage{lcsys} 

\colorlet{subsectioncolor}{.} 
\definecolor{subsectioncolor}{RGB}{39,94,77} 

\titlesc{Partitioning and Observability in Linear Systems via Submodular Optimization}

\author{Mohamad H. Kazm$\text{a}^{\diamond}$, \textit{Graduate Student Member, IEEE} and Ahmad F. Taha, \textit{Member, IEEE}  
	\thanks{
		$^\diamond$Corresponding author. This work is supported by National Science Foundation under Grants 2152450 and 2151571. The authors are with the Civil $\&$ Environmental Engineering and Electrical $\&$ Computer Engineering Departments, Vanderbilt University, 2201 West End Ave, Nashville, Tennessee 37235. Emails: mohamad.h.kazma@vanderbilt.edu, ahmad.taha@vanderbilt.edu.} 
}

\begin{document}
\newdimen\origiwspc%
\newdimen\origiwstr%
\origiwspc=\fontdimen2\font
\origiwstr=\fontdimen3\font

\maketitle
\thispagestyle{headings} 

\pagestyle{plain}

\markboth{IEEE Transactions on Automatic Control, 2026. DOI: 10.1109/TAC.2026.3676329}{}
\begin{abstract}
Network partitioning has gained recent attention as a pathway to enable decentralized operation and control in large-scale systems. This paper addresses the interplay between partitioning, observability, and sensor placement (SP) in dynamic networks. The problem, being computationally intractable at scale, {is a largely unexplored, open problem in the literature.} To that end, the paper's objective is designing scalable partitioning of linear systems while maximizing observability metrics of the subsystems. We show that the partitioning problem can be posed as a submodular maximization problem---and the SP problem can subsequently be solved over the partitioned network. Consequently, theoretical bounds are derived to compare observability metrics of the original network with those of the resulting partitions, highlighting the impact of partitioning on system observability. Case studies on networks of varying sizes corroborate the derived theoretical bounds. 
\end{abstract}
\begin{IEEEkeywords}
	Observability, network partitioning, submodular maximization, matroid constraints, multilinear extensions, sensor placement.
\end{IEEEkeywords}
\section{Introduction and Contributions}\label{sec:Intro}
\introstart{N}{etwork} partitioning, the process of segmenting a system into finite sub-networks, has gained popularity in the past decade for both static and dynamic networks. Its advantages are numerous. Some of these merits are: reducing computational complexity, identifying community structures, performing localized sensing and decentralized feedback control, reducing communication costs, improving network resilience, and minimizing cascading failures. Partitioning is not merely a theoretical concept---it has been studied in a variety of dynamic network applications (epidemics, opinion dynamics, power grids, and water systems). 

Relevant to dynamic network partitioning is observability---a fundamental system-theoretic notion that enables one to infer internal system states from some output measurements. Observing the full state is limited by the number and configuration of sensors within a system. In particular, for large and complex systems, such as cyber-physical systems~\cite{Barboni2020} and coupled biological reaction networks~\cite{Boccaletti2006}, state reconstruction is constrained by sensing and its allocation costs. As a result, one potential goal of partitioning is to identify observable subsystems (i.e., partitions) that preserve overall system observability, thereby enabling scalable localized/distributed state-estimation strategies developed for large-scale systems.

In this paper, we consider the problem of network partitioning for linear systems such that a variety of the subsystems' observability metrics are optimized. This is followed by performing localized sensor placement (SP) via a submodular set optimization approach. We note that there is no literature on studying the interplay between observability, partitioning, and SP---and there are no known methods solving this observability-based partitioning problem. 
The benefit of this approach is threefold: \textit{(i)} to identify observable interconnected subsystems, \textit{(ii)} to enable localized/decentralized state-estimation strategies, and \textit{(iii)} to yield a more scalable solution for the SP problem, which remains a subject of ongoing research due to its combinatorial properties.

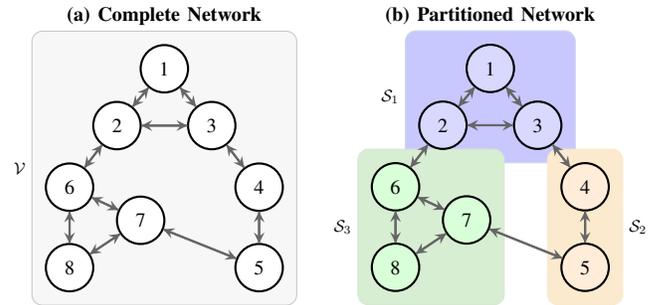
\begin{figure}[t]
	\vspace{-0.18cm}
	\centering
	\captionsetup[subfloat]{labelformat=empty} 
	\makebox[0.2\textwidth]{ 
	\subfloat[\label{fig:prop_part_1}]{
		\begin{minipage}{0.3\textwidth} 
			\centering
			\begin{tikzpicture}[scale=0.63, transform shape]
				\tikzstyle{node}=[circle, draw=black, thick, fill=white, minimum size=1cm, font=\large]
				\tikzstyle{arrow}=[->, thick, >=stealth, color=black!60]
				\node[node] (x1) at (0,2.2) {1};
				\node[node] (x2) at (-1,1) {2};
				\node[node] (x3) at (1,1) {3};
				\node[node] (x4) at (2,-0.3) {4};
				\node[node] (x5) at (2,-2) {5};
				\node[node] (x6) at (-2,-0.3) {6};
				\node[node] (x7) at (-0.5,-1) {7};
				\node[node] (x8) at (-2,-2) {8};
				
				\draw[arrow] (x2) -- (x6);
				\draw[arrow] (x6) -- (x2);
				\draw[arrow] (x2) -- (x1);
				\draw[arrow] (x1) -- (x2);
				\draw[arrow] (x2) -- (x3);
				\draw[arrow] (x3) -- (x2);
				\draw[arrow] (x3) -- (x1);
				\draw[arrow] (x1) -- (x3);
				\draw[arrow] (x3) -- (x4);
				\draw[arrow] (x4) -- (x3);
				\draw[arrow] (x4) -- (x5);
				\draw[arrow] (x5) -- (x4);
				\draw[arrow] (x6) -- (x7);
				\draw[arrow] (x7) -- (x6);
				\draw[arrow] (x7) -- (x8);
				\draw[arrow] (x8) -- (x7);
				\draw[arrow] (x8) -- (x6);
				\draw[arrow] (x6) -- (x8);
				\draw[arrow] (x5) -- (x7);
				\draw[arrow] (x7) -- (x5);
				
				\node[above=0.5cm] at (0,2.5) {\textbf{\large (a) Complete System}};
				\begin{scope}[on background layer]
					\node[draw=Vedge, fill=Vgray, 
					rounded corners, inner sep=5pt, fit=(x1) (x2) (x3) (x4) (x5) (x6) (x7) (x8), label=left:\(\mc{V}\)] {};
				\end{scope}
			\end{tikzpicture}
		\end{minipage}}
		\hspace{-1.2cm}
		\subfloat[\label{fig:prop_part_2}]{
		\begin{minipage}{0.3\textwidth} 
			\centering
			\begin{tikzpicture}[scale=0.63, transform shape]
				\tikzstyle{node}=[circle, draw=black, thick, fill=white, minimum size=1cm, font=\large]
				\tikzstyle{colorednode1}=[circle, draw=black, thick, fill=blue!15, minimum size=1cm, font=\large]
				\tikzstyle{colorednode2}=[circle, draw=black, thick, fill=orange!15, minimum size=1cm, font=\large]
				\tikzstyle{colorednode3}=[circle, draw=black, thick, fill=green!15, minimum size=1cm, font=\large]
				\tikzstyle{arrow}=[->, thick, >=stealth, color=black!60]
				
				\node[colorednode1] (x1) at (0,2.2) {1};
				\node[colorednode1] (x2) at (-1,1) {2};
				\node[colorednode1] (x3) at (1,1) {3};
				\node[colorednode2] (x4) at (2,-0.3) {4};
				\node[colorednode2] (x5) at (2,-2) {5};
				\node[colorednode3] (x6) at (-2,-0.3) {6};
				\node[colorednode3] (x7) at (-0.5,-1) {7};
				\node[colorednode3] (x8) at (-2,-2) {8};
				
				\draw[arrow] (x2) -- (x6);
				\draw[arrow] (x6) -- (x2);
				\draw[arrow] (x2) -- (x1);
				\draw[arrow] (x1) -- (x2);
				\draw[arrow] (x2) -- (x3);
				\draw[arrow] (x3) -- (x2);
				\draw[arrow] (x3) -- (x1);
				\draw[arrow] (x1) -- (x3);
				\draw[arrow] (x3) -- (x4);
				\draw[arrow] (x4) -- (x3);
				\draw[arrow] (x4) -- (x5);
				\draw[arrow] (x5) -- (x4);
				\draw[arrow] (x6) -- (x7);
				\draw[arrow] (x7) -- (x6);
				\draw[arrow] (x7) -- (x8);
				\draw[arrow] (x8) -- (x7);
				\draw[arrow] (x8) -- (x6);
				\draw[arrow] (x6) -- (x8);
				\draw[arrow] (x5) -- (x7);
				\draw[arrow] (x7) -- (x5);
				
				\node[above=0.5cm] at (0,2.5) {\textbf{\large (b) Partitioned Subsystems}};

				\begin{scope}[on background layer]
					\node[draw=Sub1edge, fill=Sub1purple, 
					rounded corners, inner sep=5pt, fit=(x1) (x2) (x3), label=left:\(\mc{S}_1\)] {};
					
					\node[draw=Sub2edge, fill=Sub2yellow, 
					rounded corners, inner sep=5pt, fit=(x4) (x5), label=right:\(\mc{S}_2\)] {};
					
					\node[draw=Sub3edge, fill=Sub3green, 
					rounded corners, inner sep=5pt, fit=(x6) (x7) (x8), label=left:\(\mc{S}_3\)] {};
				\end{scope}
			\end{tikzpicture}
		\end{minipage}}
	} 
	\vspace{-0.2cm}
	\caption{(a) A system of measured internal state, $ v\in \mc{V}$ (measurable space), depicting interactions of a dynamical system and (b) the subsequent subsystems $\mc{S}_i \subseteq \mc{V}$. The nodes represent the system states and edges represent the internal state connections. The colored boxes represent the nodes that belong to a particular subsystem. The interactions between the subsystems remain after system partitioning.}\label{fig:subsystem}
\end{figure}
\setlength{\textfloatsep}{2pt}
The SP problem has been thoroughly studied in the literature~\cite{Joshi2009,Summers2016,Tzoumas2016,Summers2019,Zhou2019,Vinod2022}. One approach follows from~\cite{Summers2016} that formulates the actuator placement, dual to the SP problem, in the context of submodular set optimization. This setting is characterized by diminishing returns properties which provides provable performance guarantees that allow one to solve the combinatorial SP problem in polynomial time~\cite{Guo2021}. However, the applicability of polynomial-time algorithms, such as the simple greedy, still yields scalability challenges for large-scale systems and real-time sensor scheduling. To this end, in this paper we approach alleviating the computational burden of the SP problem in large-scale linear time-invariant (LTI) systems by identifying subsystems within the original dynamical system as described in Fig.~\ref{fig:subsystem}. This allows the exploitation of more advanced algorithms that leverage parallelization over the subsystems. Specifically, this paper studies the fundamental problem of partitioning or clustering of a dynamic linear network to perform two key functions: \textit{(i)} maximizing the partitions' observability metrics while \textit{(ii)} ensuring optimized SP for overall system observability. This also enables, as broader impact, localized state estimation for the partitioned interconnected subsystems. We perform the above two key functions via a specific set of techniques, grounded in submodular set function optimization using matroid constraints. In what follows, we briefly highlight the rich literature of network partitioning, SP, and their overlap. 

\parstartc{Clustering versus partitioning} We note here that there is somewhat of an ambiguity in the literature of partitioning/clustering, where there is no consensus among the practical interpretation of these two words. While some studies define partitioning through dissecting a network into $k$ partitions that are no longer connected~\cite{Raak2016,She2020}, others maintain the connectivity in the network~\cite{Ishizaki2014a,Ananduta2021,Mattioni2022}. In this paper, we consider the problem of partitioning a system into clustered communities, referred to as subsystems, where the interconnections between the measurable states of the subsystems are maintained. Notice from Fig.~\ref{fig:prop_part_2}, the connections (arrows) between subsystems $\mc{S}_1,\mc{S}_2,\mc{S}_3$ remain intact after the partitioning.

\parquestc{Why is partitioning needed} The clustering of nodes and network states is becoming increasingly relevant for the scalability of distributed control methods~\cite{Chanfreut2021}. In a dual concept to observability, a submodular approach towards islanding of multi-agent networks based on consensus-based control is formulated in~\cite{Cheng2024}. The controlled islanding or partitioning of a power grid is formulated to reduce the effects of cascading failure events. Furthermore, while studying control under multi-agent networks and coupled systems, the \textit{herdability}\footnote{Herdability is a relaxed notion of controllability that indicates the ability to steer system states to a specific state-space subset.} of a set of leader nodes in a cluster is introduced in recent literature~\cite{She2020,DePasquale2023}. The dynamical network is \textit{equitable} partitioned to characterize a controllable subspace that renders the system herdable. In~\cite{Yu2018,Yu2019}, the local subsystem dynamics of a discrete LTI system are identified by leveraging local subsystem input and output signals. An optimization-based partitioning algorithm is proposed in~\cite{Arastou2025}, where local controllability is imposed as a constraint through the rank of the controllability matrix. The method demonstrates that partitioning can reduce computational costs for local controllers, thereby improving decentralized control performance. Readers are referred to~\cite{Chanfreut2021} for a survey on clustering strategies for scalable distributed control, and to~\cite{Riccardi2025} for a recent survey on partitioning in model predictive control; clustering remains a topic of interest.

The partitioning of complex systems provides an essential decomposition for effective and scalable decentralized state estimation in large-scale dynamical systems~\cite{Daoutidis2019}.
For system monitoring under limited computational and sensing resources, the authors in~\cite{Niazi2023} propose a clustering-based approach that constructs an aggregated observer for large-scale linear flow networks. The method designs an observer based on the average state of each cluster while achieving computational tractability over the clusters. However, it does not explicitly ensure or preserve observability within individual clusters during the clustering process. The authors in~\cite{Corah2018} propose a submodular distributed sensor network planning framework. The method randomly partitions network agents to enable efficient distributed decision-making while maintaining constant-factor bound guarantee. A partitioning of a sensing network based on stability margin is introduced in~\cite{Hamdipoor2019}. In~\cite{Rezazadeh2021}, the optimal strategy of an agent in communication within a multi-agent network is solved via a continuous submodular maximization of a shared utility function under disjoint strategy sets. 

\parstartc{Paper contributions} As compared to the aforementioned literature, this paper partitions an LTI system into $k$ interconnected subsystems by maximizing an observability-based metric/measure of the subsystem dynamics. Observing internal states from external measurements implies that there exists a coupling network between internal states, this enables leveraging observability for structural inference~\cite{Succar2025}. Building on this notion, we exploit the internal state connections, quantified using each subsystem’s observability Gramian, to identify subsystems that partition a system while maximizing an observability-based metric. The SP problem, then solved for the partitions, is related to the subsystem observability properties of the dynamical system; it ensures that state information is retained within each cluster while maximizing overall system observability. The observability-based partitioning problem, and the interplay between partitioning and observability in linear systems, has not been thoroughly investigated and remains an open problem to the best of the authors’ knowledge; this paper aims to fill this gap.
\begin{itemize}
    \item We introduce a system partitioning framework that maximizes observability-based metrics of LTI subsystem dynamics, while considering overall system's observability. The partitioning problem is formulated as a submodular set function maximization problem under a partition matroid which dictates the allocation of measurable states to partitions that yield maximal observability gain. This makes the otherwise combinatorial partitioning problem solvable in polynomial time. In doing so, we also show that the objective function representing the observability measures of the subsystem dynamics retains the modularity, submodularity, and monotonicity properties of the original observability-based measures under no system partitioning.
	\item We extend the observability-based SP problem to partitioned systems by solving it under a partition matroid constraint which represents the sensors that can be allocated to each partition. We show that the submodular properties of parameterized observability measures are retained for a partitioned system, enabling the use of greedy algorithms.
	\item We efficiently solve the optimal partitioning by utilizing a continuous greedy algorithm. A multilinear relaxation is utilized to enable greedy selection under a matroid constraint while achieving a $(1 - 1/e)$ performance guarantee.
	\item We provide theoretical bounds that characterize the relationship between the observability measures of the global system and those of the sum of local subsystems. Such bounds quantify performance loss—or lack thereof—when solving the SP problem in the partitioned versus unpartitioned system.
    \item We validate the proposed partition and SP framework on two combustion reaction networks. The results empirically validate the theoretical bounds for the $\log\!\det$ metric under varying partition numbers, and sensor configurations.
\end{itemize}

\parstartc{Notation} Let $\mathbb{N}$, $\R$, and $\R_{\geq 0}$ denote the set of natural, real, and non-negative real numbers. Let $\Rn{n}$ and $\Rn{n\times m}$ denote the set of real-valued column vectors of size $n$, and $n$-by-$m$ real matrices. The operators $\log\!\det(\m{A})$, $\trace(\m{A})$ and $\rank{(\m{A})}$ return the logarithmic determinant, trace and rank of matrix $\m{A}\in \Rn{n\times n}$. For any two matrices $\m{A}$ and $\m{B} \in \Rn{n\times n}$, the notation $\{\m{A}, \m{B}\}$ represents $\left[\m{A}^{\top} \m{B}^{\top}\right]^{\top}$. For a symmetric matrix $\m{A}$, the symbol $\m{A} \succeq 0$ denotes a positive semidefinite (PSD) matrix and $\m{A} \succ 0$ denotes a positive definite (PD) matrix. The symbol $\emptyset$ denotes the empty set. The cardinality of a set $\mc{V}$ is denoted by $|\mc{V}|$. The symbols $\subseteq$, $\cap$, $\cup$ and $\backslash$ denote set inclusion (subset), set intersection, set union, and set difference. The notation $2^{\mc{V}}$ denotes the power set of $\mc{V}$ (the set of all subsets of $\mc{V}$). For $v \in \mc{V}$ denoting a single element from a set, the notation $\mc{V} \cup \{v\}$ represents adding element $v$ to $\mc{V}$. The Cartesian product, $\mc{C} \times \mc{V}$, of two sets $\mc{C}$ and $\mc{V}$, represents all ordered pairs $(i,v)$ where $i \in \mc{C}$ and $v \in \mc{V}$.

\parstartc{Paper organization} The structure of the paper is as follows. Section~\ref{sec:problemformulation} presents the preliminaries and the problem formulation. Section~\ref{sec:main-reformulation} introduces the partitioning problem under a partition matroid constraint. Section~\ref{sec:main} provides evidence for the submodularity of observability-based measures in a partitioned system. The SP problem for a partitioned system is formulated in Section~\ref{sec:SNS}. The multilinear extension is introduced in~Section~\ref{sec:multilinear}. Numerical results are presented in Section~\ref{sec:simulation}. Section~\ref{sec:conclusion} concludes this paper.
\section{Preliminaries and Problem Formulation}\label{sec:problemformulation}
\subsection{Properties of submodular set functions}\label{subsec:setfunctions}
Consider a ground set $\mc{V}$ and its power set $2^{\mc{V}}$. A set function can be written as $f: 2^{\mc{V}}\rightarrow \mathbb{R}_{\geq 0}$. If $f$ exhibits a diminishing-returns property, it is said to be submodular (Definition~\ref{def:modular_submodular}). This means that adding a single element to an argument of $f$ has an incremental value no less than that of adding the same element to a superset of that argument.
\begin{mydef}\textit{(\hspace{-0.012cm}Submodularity~\cite{Bach2010})}\label{def:modular_submodular}
	The set function $f:2^{\mc{V}}\rightarrow \mbb{R}_{\geq 0}$ is \textit{modular} if there exists a weight function $w:\mc{V}\rightarrow \mbb{R}_{\geq 0}$ such that for any $\mc{S}\subseteq\mc{V}$, it holds that 
	$ f(\mc{S}) = w(\emptyset) + \sum_{s\in\mc{S}} w(s),$
	and $f$ is considered \textit{submodular} if for any $\mc{A} \subseteq \mc{B}\subseteq\mc{V}$ and for any $s\notin\mc{B}$, the following holds
	\begin{align}\label{eq:submodular_def}
		f(\mc{A}\cup\{s\}) - f(\mc{A})\geq  f(\mc{B}\cup\{s\}) - f(\mc{B}).
	\end{align}
\end{mydef}

The set function $f$ is said to be supermodular if the reverse inequality in~\eqref{eq:submodular_def} holds true for all $s\notin\mc{B}$. We say that a set function is normalized if $f(\emptyset)=0$. Furthermore, the set function $f$ is monotone increasing if for any $\mc{A}\subseteq\mc{B}\subseteq\mc{V}$, we have $	f(\mc{B})\geq 	f(\mc{A})$. A set function $f$ that is submodular, monotone increasing and normalized, is termed a \textit{polymatroid} set function~\cite{Bilmes2022}. Optimization problems that incorporate submodular objective functions are often constrained on an admissible subset $\mc{S}\subseteq \mc{V}$ of the ground set. Constraints that arise in the context of submodular set optimization are matroid constraints. The following provides their definition.
\begin{mydef}\textit{(\hspace{-0.012cm}Matroid Constraints~\cite{Fujishige2013})}\label{def:matroid-constraint}
	Let $\mc{I}$ be a collection of subsets of $\mc{V}$. A constraint can then be represented by the tuple $\mc{M} = (\mc{V}, \mc{I})$. The tuple $\mc{M}$ is a matroid constraint if the following hold true: \textit{(i)} $\emptyset \in \mc{I}$; \textit{(ii)} $\mc{A} \subseteq \mc{B} \in \mc{I} \Rightarrow \mc{A} \in \mc{I}$; and \textit{(iii)} if $\mc{A}, \mc{B} \in \mc{I}$ and $|\mc{A}| < |\mc{B}|$, then there exists $s \in \mc{B} \backslash \mc{A}$ such that $\mc{A} \cup \{s\} \in \mc{I}$.
\end{mydef}

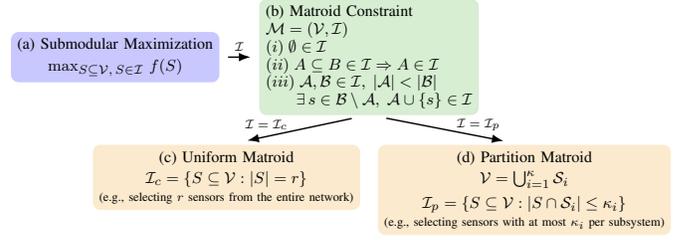
\begin{figure}[t]
	\vspace{+0.2cm}
	\centering
	\resizebox{0.99\linewidth}{!}{%
		\begin{tikzpicture}[node distance=0.5cm, every node/.style={align=center, font=\small}]
			\node[draw=Sub1purple, rectangle, rounded corners, fill=Sub1purple] (opt) {(a) Submodular Maximization\\[0.5ex] \(\max_{S\subseteq \mc{V},\, S \in \mc{I}}\, f(S)\)};

			\node[draw=Sub3edge, fill=Sub3green, rectangle, rounded corners,  right=of opt,  xshift=0.3cm, align=left] (matroid) {(b) Matroid Constraint\\ \(\mc{M}=(\mc{V},\mc{I})\)\\ 
				\((i)\; \emptyset \in \mc{I}\)\\ 
				\((ii)\; A \subseteq B \in \mc{I} \Rightarrow A \in \mc{I}\)\\ 
				\((iii)\; \mc{A}, \mc{B} \in \mc{I},\; |\mc{A}|<|\mc{B}|\)\\ 
				\hspace{0.5cm} \(\exists\, s \in \mc{B} \setminus \mc{A},\; \mc{A} \cup \{s\} \in \mc{I}\)};

			\node[draw=Sub2edge, rectangle, rounded corners, fill=Sub2yellow, below left=0.6cm and -2cm of matroid] (uniform) {(c) Uniform Matroid\\[0.5ex] \(\mc{I}_c=\{S\subseteq \mc{V}: |S|=r\}\) \\ {\scriptsize(e.g., selecting $r$ sensors from the entire network)} };

			\node[draw=Sub2edge, rectangle, rounded corners, fill=Sub2yellow, below right=0.6cm and -2cm of matroid] (partition) {(d) Partition Matroid\\[0.5ex] \(\mc{V}=\bigcup_{i=1}^{\kappa}\mc{S}_i \vspace{+0.05cm}\)\\ \(\mc{I}_{p}=\{S\subseteq \mc{V}: |S\cap \mc{S}_i|\leq \kappa_i\}\) 
			\\ {\scriptsize(e.g., selecting sensors with at most $\kappa_i$ per subsystem)}};
			
			\draw[-{Latex}
			, thick,shorten <=5pt, shorten >=5pt] (opt) -- (matroid) node[midway, above] {\scriptsize$\mc{I}$};

			\draw[-{Latex}, thick, shorten <=10pt, shorten >=12pt] 
			(matroid.south) -- (uniform.north) node[midway, left, xshift=-2pt, yshift=3pt] {\scriptsize$\mc{I} = \mc{I}_c$}; 
			
			\draw[-{Latex}, thick, shorten <=10pt, shorten >=12pt] 
			(matroid.south) -- (partition.north) node[midway, right, xshift=2pt, yshift=3pt] {\scriptsize$\mc{I} = \mc{I}_p$};
		\end{tikzpicture}%
	}
	\vspace{-0.5cm}
	\caption{Matroid Constraints: For any (a) submodular maximization problem, a matroid constraint (b) satisfies: \textit{(i)} the null property, \textit{(ii)} the heredity property and \textit{(iii)} augmentation property. Two common matroid constraints are the (c) uniform matroid $\mc{I}_c$ and (d) partition matroid $\mc{I}_p$.}
	\label{fig:matroid2}
\end{figure}

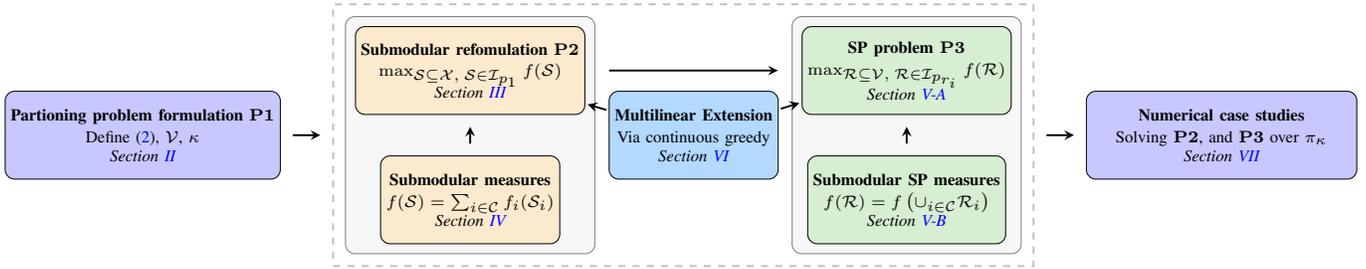
\begin{figure*}[t]
	\centering
	\scriptsize
	\resizebox{1\textwidth}{!}{
		\begin{tikzpicture}[
			node distance=0.6cm and 0.95cm,
			box/.style={draw, rounded corners, fill=blue!5, minimum height=1.3cm, minimum width=2.5cm, align=center},
			arrow/.style={->,>=stealth, thick,shorten <=5pt, shorten >=5pt},
			label/.style={font=\normal, align=center}
			]
			
			\node[box, fill=Sub2yellow] (p2) {
				\textbf{Submodular reformulation $\mb{P2}$} \\[1pt]
				$\max_{\mc{S}\subseteq\mc{X},\; \mc{S} \in \mc{I}_{p_1}}  f(\mc{S})$ \\
				\textit{Section~\ref{sec:main-reformulation}}
			};
			
			\node[box, fill=Sub2yellow, below=of p2] (submod) {
				\textbf{Submodular measures} \\[1pt]
				$f(\mc{S}) = \sum_{i \in \mc{C}} f_i(\mc{S}_i)$\\
				\textit{Section~\ref{sec:main}}
			};
			
			\begin{scope}[on background layer]
				\node[
				draw=gray, 
				fill={rgb,255:red,247; green,247; blue,247}, 
				rounded corners, 
				inner sep=4pt, 
				fit= (p2) (submod),
				] (P2) {};
			\end{scope}

			\node[box, fill=Sub1purple, left=of P2] (setup) {
				\textbf{Partitioning problem formulation $\mb{P1}$}\\[1pt]
				 Define \eqref{eq:model_DT}, $\mc{V}$, $\kappa$\\
				\textit{Section~\ref{sec:problemformulation}}
			};
			
			\node[box, fill=Sub3green,anchor=north west] at ([xshift=3.2cm]p2.north east) (p3) {
				\textbf{SP problem $\mb{P3}$} \\[1pt]
				$\max_{\mc{R}\subseteq\mc{V},\;\mc{R}\in\mc{I}_{p_{r_i}}} f(\mc{R})$\\
				\textit{Section~\ref{subsec:SNS-1}}
			};
			
			\node[box, fill=Sub3green, below=of p3] (submod2) {
				\textbf{Submodular SP measures} \\[1pt]
				$f(\mc{R}) =f\left(\cup_{i\in\mc{C}} \mc{R}_{i}\right)$\\
				\textit{Section~\ref{subsec:SNS-2}}
			};
			
			\begin{scope}[on background layer]
				\node[
				draw=gray, 
				fill={rgb,255:red,247; green,247; blue,247}, 
				rounded corners, 
				inner sep=4pt, 
				fit= (p3) (submod2),
				] (P3) {};
			\end{scope}
			
			\node[box, fill=LightBlue, align=center] 
			at ($(P2.center)!0.512!(P3.center)$) 
			(multi) {
				\textbf{Multilinear Extension}\\[1pt]
				Via continuous greedy\\
				\textit{Section~\ref{sec:multilinear}}
			};
			
			\node[box, fill=Sub1purple, right=of P3,minimum width=4cm] (sim) {
				\textbf{Numerical case studies}\\[1pt]
				Solving $\mb{P2}$, and $\mb{P3}$ over $\pi_{\kappa}$\\
				\textit{Section~\ref{sec:simulation}}
			};
			\node[draw=gray!50, dashed, thick, fit=(P2)(P3), inner sep=5pt, name=PartitionToSP] {};
			
			\draw[arrow] (setup) -- (PartitionToSP);
			\draw[arrow] (submod) -- (p2);
			\draw[arrow] (submod2) -- (p3);
			\draw[arrow,shorten <=10pt, shorten >=10pt] (p2) -- (p3);
			\draw[arrow] (PartitionToSP) -- (sim);
			\draw[arrow,shorten <=1pt, shorten >=1pt] (multi) -- (p2);
			\draw[arrow,shorten <=1pt, shorten >=1pt] (multi) -- (p3);
			
		\end{tikzpicture}
	}
	\vspace{-0.2cm}
	\caption{Overview of the system partitioning problem under $\mb{P2}$ (Sections~\ref{sec:problemformulation}–\ref{sec:main}), followed by SP under $\mb{P3}$ (Section~\ref{sec:SNS}). The continuous greedy algorithm can be utilized to solve $\mb{P2}$ and $\mb{P3}$ (Section~\ref{sec:multilinear}) for an LTI system such as a linearized combustion reaction network (Section~\ref{sec:simulation}).}
	\label{fig:framework}
	\vspace{-0.4cm}
\end{figure*}
\setlength{\textfloatsep}{2pt}

 Note that a matroid is an abstraction of linear independence and dependence structure of a set of columns of a matrix. A matroid constraint $\mc{M}_c = (\mc{V},\mc{I}_c)$ that allows for SP in dynamical networks is the uniform cardinality constraint $\mc{I}_c=\{ \mc{S} \subseteq \mc{V}: \abs{\mc{S}} = r\}$ for some $r$ that represents the feasible number of sensing nodes. A relevant matroid that considers multi-network systems is the partition matroid $\mc{M}_p = (\mc{V},\mc{I}_{p_{\kappa_i}})$, where $\mc{I}_{p_{\kappa_i}}=\{ \mc{S} \subseteq \mc{V}: |\mc{S} \cap \mc{S}_i| \leq \kappa_i\}$.

 The set $\mc{V}$ is partitioned into $\pi_{\kappa} = \left\{\mc{S}_1, \ldots, \mc{S}_\kappa  \right\}$ blocks such that $\mc{S}_i \cap \mc{S}_j=\emptyset$ for $i \neq j$. This means that a partition matroid is one where the ground set is partitioned into blocks, such that an independent set $\mc{S}$ intersects the block by no more than a specific limit~\cite{Wang2021c}. Formulating a submodular maximization problem under a matroid constraint allows for proven theoretical optimality guarantees~\cite{Calinescu2011}. In particular, uniform matroids (Fig.~\ref{fig:matroid2}) apply when the number of allowable selections is limited to a fixed cardinality, as in the case of SP problems, while partition matroids (Fig.~\ref{fig:matroid2}) apply when the ground set $\mc{V}$ is divided into disjoint blocks $\mc{S}_i$, with at most $\kappa_i$ elements selected from each block. In the context of system observability, these constraints arise when assigning a limited number of sensors across different subsystems. The resulting optimization problem, therefore, seeks to maximize a monotone submodular function subject to a partition matroid constraint rather than a uniform constraint. Similar to SP under uniform matroids, such a formulation (under partitioning) enables leveraging different greedy algorithms while ensuring performance guarantees under submodularity. The next section presents the system partitioning problem formulation.
\subsection{Problem formulation: partitioning linear systems}\label{susec:problemformulation}
We consider a LTI discrete dynamical system with a linear measurement model given as
\begin{equation}\label{eq:model_DT}
	\m{x}[{k+1}] = \m{A}\m{x}[k] + \m{B}\m{u}[k], \;\;\;\;\; \quad
	\m y[k] =  \m{C} \m{x}[{k}],
\end{equation}
where $k\in\mbb{N}$ refers to the discrete-time index, such that vector $\m{x}[k]  \in \Rn{n_x}$ represents the dynamic states, vector $\m{u}[k] \in  \Rn{n_u}$ represents the input state vector, and vector $\m{y}[k] \in \Rn{n_y}$ the output measurements. Matrices $\m{A}\in \Rn{n_x \times n_x}$, $\m{B}\in \Rn{n_x \times n_u}$, and $\m{C}\in \Rn{n_y \times n_x}$ are the LTI state-space matrices. 

Considering the above LTI system~\eqref{eq:model_DT}, the SP problem can be solved by first identifying $\kappa$ number of subsystems. This entails optimally partitioning the system dynamics while maximizing an observability-based metric for each of the subsystems and then solving the SP problem under a partition matroid rather than a cardinality constraint. Such a framework renders the problem computationally feasible given that the combinatorial problem is exponentially constrained by the largest number of states within a subsystem. For an unpartitioned system this is equal to the full measurement state-space dimension $n_y$. The partitioning problem can be summarized as follows: \textit{(i)} given the set of measurable state-space variables, represented by the ground set $\mc{V}$ covering the full state-space of the LTI system~\eqref{eq:model_DT}, the objective is to partition $\mc{V}$ into $\kappa$ disjoint subsystems $\mc{S}_i$, each representing a group of states; \textit{(ii)} by partitioning based on maximizing an observability-based metric associated with each subsystem, the resulting partitioning yields a decomposition in which each subsystem achieves an optimal observability measure. Accordingly, the measurable state partitioning problem ($\mb{P1}$) of system~\eqref{eq:model_DT} is defined as follows.
\begin{problem}\label{problem:P1}
	Let the ground set $\mc{V} := \{ v\in\mbb{N}\,|\,0 < v \leq n_y\}$ represent the measurable state-space of cardinality $|\mc{V}| = n_y$ and $f\left(\mc{S}_1, \ldots, \mc{S}_\kappa \right): 2^{\mc{V}} \rightarrow \mbb{R}_{\geq 0}$ is a sum of polymatroid functions ${f}_{i}: 2^{\mc{V}} \rightarrow \mbb{R}_{\geq 0} \; \forall\; i \in \mc{C} = \{1, \dots, \kappa\}$. Then, the partitioning into $\kappa$ subsystems is as follows
\begin{subequations}\label{eq:P1}
	\begin{align}
		\llap{$\mb{P1}$:}
		\; \; \; \; \max _{\mc{S}_1, \ldots, \mc{S}_{\kappa}}& \quad f\left(\mc{S}_1, \ldots, \mc{S}_{\kappa}\right) = \sum_{i=1}^{\kappa} f_i\left({\mc{S}_{i}}\right),\label{eq:P1_obj}\\
		\quad \st& \quad \bigcup_{i=1}^{\kappa} \mc{S}_{i}=\mc{V}, \; \mc{S}_{i} \cap \mc{S}_{j}=\emptyset \;\; \forall \; i \neq j,\label{eq:P1_contraint}
	\end{align}
\end{subequations}
\end{problem}
where $f^{*}_{\pi_{\kappa}} := \max_{\mc{S}_1, \ldots, \mc{S}_{\kappa}} f\left(\mc{S}_1, \ldots, \mc{S}_{\kappa}\right)$. We assume here that all system states can be observed via a to-be-placed sensor. The partitions $\pi_{\kappa} = \left\{\mc{S}_1, \ldots, \mc{S}_{\kappa}\right\}$ are non-overlapping partitions of ground set $\mc{V}$ such that the union of all partitions covers $\cup_{i=1}^{\kappa} \mc{S}_{i}$ is equal to $\mc{V}$, and the intersection of two partitions is the null set for all $i\neq j$, i.e., $\mc{S}_{i} \cap \mc{S}_{j}=\emptyset$. The resulting system partition $\mc{\pi}_{\kappa}$ is $\kappa$ subsystems of size $\kappa_i$ for each $i\in \mc{C}$. Here $\mc{C}$ is the index set of the partitions. Observe that the constraints in $\mb{P1}$ do not hold the properties of a matroid constraint (Definition~\ref{def:matroid-constraint}). The feasible sets are tuples of disjoint sets \(\mc{S}_i\) such that their union is equal to the ground set \(\mc{V}\), which does not satisfy the \textit{(ii)} hereditary and \textit{(iii)} augmentation properties of a matroid.

Thus to enable the use of greedy (and continuous greedy) algorithms with provable theoretical guarantees on the solution's optimality, we rewrite $\mb{P1}$. Note that $\mb{P1}$ can be considered a submodular welfare problem in the field of discrete optimization~\cite{Vondrak2008a}. Thus, we reformulate the original partitioning problem $\mb{P1}$ into a submodular maximization subject to a matroid constraint, denoted as $\mb{P2}$. As illustrated in Fig.~\ref{fig:framework}, the partitioning problem $\mb{P2}$ and the following SP problem exhibit and retain submodular properties of the original system under well-known observability measures. A solution under the multilinear extension of $\mb{P2}$ yields such theoretical guarantees. The ensuing section formulates the system partitioning problem in the context of submodular maximization under a matroid constraint.
\section{Reformulation using Submodular Set Function Optimization}\label{sec:main-reformulation}
To exploit submodular maximization via greedy algorithms, the partitioning problem $\mb{P1}$ can be reformulated as a submodular maximization problem under a partition matroid constraint. This entails creating $\kappa$ copies of the measurable states in the ground set $\mc{V}$, one for each subsystem $\mc{S}_i$. The subsystems, initialized as empty sets, then \textit{compete} to acquire the states $v \in \mc{V}$, each forming a unique set of states by setting a partition matroid constraint with $\kappa_i = 1$.

\subsection{Competing subsystems for system partitioning}\label{subsec:main2-1}
We consider partitioning the measurable state-space into $\kappa$ subsystems. The subsystems $\mc{S}_i$ acquire measurable states $v \in \mc{V}$ while maximizing an observability-based measure $f_i\left({\mc{S}_{i}}\right)$ of the subsystems. Given that the partition sets are disjoint, i.e., each state $v$ is assigned to only one subsystem, then each $\mc{S}_i$ competes for each of the available states. This is similar to the welfare allocation problem, where certain resources are allocated among a number of agents while maximizing the utility of each agent for the allocated resource, such that each state is considered a resource to be claimed by at most one subsystem. As such, we reformulate $\mb{P1}$ as a submodular maximization over an extended ground set $\mc{X}:=\mc{C} \times \mc{V}$. This creates $\kappa$ copies of states $v$, one copy for each of the subsystems. Then, by considering the assignment process using a partition matroid with $\kappa_i = 1$ for all $i \in \mc{C}$, we ensure that a state $v$ is allocated to at most one subsystem. The resulting optimization problem, referred to as $\mb{P2}$, can be defined as follows.
\begin{problem}\label{problem:P2}
	Let the ground set $\mc{X}=\mc{C} \times \mc{V}$ be the Cartesian product of set $\mc{V}$ representing the measurable state-space and $\mc{C}$ the subsystem index set, where an element $(i,v) \in \mc{X}$ represents the state $v$ being assigned to subsystem $\mc{S}_i$. Let the submodular set function be defined as $f: 2^{\mc{X}} \rightarrow \mbb{R}_{\geq 0}$, then $\mb{P1}$ can be rewritten as
\begin{subequations}\label{eq:P1-2}
	\begin{align}
	   \llap{$\mb{P2}$:} \;\;\;
		\max _{\mc{S}}& \;\; f\left(\mc{S}\right) = \sum_{i \in \mc{C}} 	f_i\Bigl(\{v\in\mc{V} \mid (i,v)\in S\}\Bigr),\hspace{-0.2cm}\label{eq:P1-2-obj}\\
		\;\;\; \st& \;\left\{\mc{S} \subseteq \mc{X}: \left|\mc{S} \cap (\mc{C}\times \{v\})\right| \leq 1, \; \forall \;  v \in \mc{V} \right\}\hspace{-0.05cm}.\label{eq:P1-2-constraint}
	\end{align}
\end{subequations}
\end{problem}
where $f_{\pi^{*}_{\kappa}} = \max_{\mc{S}} f\left(\mc{S}\right).$

Problem $\mb{P2}$ dictates that each state $v$ in the ground set $\mc{V}$ is duplicated $\kappa$ times across the subsystems $\forall\; i \in \mc{C}$. The constraint $\left|\mc{S}\cap(\mc{C}\times \{v\})\right| \leq 1$ ensures that each state $v \in \mc{V}$ is assigned to at most one partition as a result of setting $\kappa_i=1$ for each of the subsystems in $\mc{\pi}_{\kappa}$. The submodular set function $f(\mc{S})$ can be defined such that every set $\mc{S} \subseteq \mc{X}$ can be written as $\mc{S}=\bigcup_{i \in \mc{C}}\left(\{i \} \times \mc{S}_i\right)$. Note that $f$, if originally submodular, retains its submodularity under a nonnegative sum of submodular functions $f_i$, where each $f_i: 2^{\mc{V}}\rightarrow\R_{\geq0}$ is a submodular observability-based measure assigned to the subsystem $\mc{S}_i$. In fact, a non-negative conic combination of submodular functions remains submodular. In the subsequent sections, a theoretical analysis regarding the submodularity of observability-based measures defined by $f(\mc{S})$ under that sum of partitioned observability measures $f_i(\mc{S}_i)$ is provided. 

The following lemma establishes that problems $\mb{P1}$ and $\mb{P2}$, for partitioning observable states $v\in\mc{V}$ according to observability-based submodular measures, are equivalent; see Fig.~\ref{fig:comparison_P1P2} for an illustrative example.

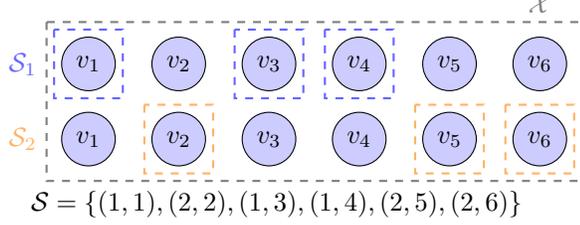
\begin{figure}[t]
	\centering
	\vspace{0.2cm}
			\centering
			\begin{flushleft}
				\textbf{(a) P1:} $\cup_{i=1}^{\kappa} \mc{S}_{i}=\mc{V},\; \mc{S}_i \cap \mc{S}_j = \emptyset$
			\end{flushleft}
			\vspace{-0.4cm}
			\begin{tikzpicture}[scale=1, transform shape, every node/.style={align=center, font=\normalsize}]
				\hspace{+0.3cm}
				\foreach \i/\x/\v in {1/0/\(v_1\), 2/1.2/\(v_3\), 3/2.4/\(v_4\), 4/3.6/\(v_2\), 5/4.8/\(v_5\), 6/6/\(v_6\)} {
					\node[draw, circle, fill=blue!20, minimum size=0.7cm] (v\i) at (\x,1) {\v};
				}
				
				\draw[dashed, thick, color=gray] ($(v1.north west)+(-0.3,0.3)$) rectangle ($(v6.south east)+(0.3,-0.3)$);
				\node[above=0.2cm of $(v6.north)!0.0!(v6.north)$,color=gray] {\(\mc{V}\)};
				
				\draw[dashed, thick, color=blue!60] ($(v1.north west)+(-0.2,0.2)$) rectangle ($(v3.south east)+(0.2,-0.2)$);
				\node[left=0.2cm of v1, blue!60] {\(\mc{S}_1\)};
				
				\draw[dashed, thick, color=orange!60] ($(v4.north west)+(-0.2,0.2)$) rectangle ($(v6.south east)+(0.2,-0.2)$);
				\node[right=0.2cm of v6, orange!60] {\(\mc{S}_2\)};
								
				\node[below=0.2cm of v2,xshift=0.3cm] {$\mc{S}_1 = \{1,3,4\}$, $\mc{S}_2 = \{2,5,6\}$};
			\end{tikzpicture}
        \vspace{-0.3cm}
        
			\centering
			\begin{flushleft}
 					\textbf{(b) P2:} $\mc{S} \subseteq \mc{X},\; \left|\mc{S} \cap (\mc{C}\!\times\!\{v\})\right|\leq 1$
			\end{flushleft}
			\vspace{-0.4cm}
			\begin{tikzpicture}[scale=1, transform shape, every node/.style={align=center, font=\normalsize}]
							
				\foreach \i/\x/\v in {1/0/\(v_1\), 2/1.2/\(v_2\), 3/2.4/\(v_3\), 4/3.6/\(v_4\), 5/4.8/\(v_5\), 6/6/\(v_6\)} {
					\node[draw, circle, fill=blue!20, minimum size=0.7cm] (v\i t) at (\x,1) {\v};
					\node[draw, circle, fill=blue!20, minimum size=0.7cm] (v\i b) at (\x,0) {\v};
				}
				
				\draw[dashed, thick, color=blue!60] ($(v1t.north west)+(-0.2,0.2)$) rectangle ($(v1t.south east)+(0.2,-0.2)$);
				\draw[dashed, thick, color=blue!60] ($(v3t.north west)+(-0.2,0.2)$) rectangle ($(v3t.south east)+(0.2,-0.2)$);
				\draw[dashed, thick, color=blue!60] ($(v4t.north west)+(-0.2,0.2)$) rectangle ($(v4t.south east)+(0.2,-0.2)$);
				
				\draw[dashed, thick, color=orange!60] ($(v2b.north west)+(-0.2,0.2)$) rectangle ($(v2b.south east)+(0.2,-0.2)$);
				\draw[dashed, thick, color=orange!60] ($(v5b.north west)+(-0.2,0.2)$) rectangle ($(v5b.south east)+(0.2,-0.2)$);
				\draw[dashed, thick, color=orange!60] ($(v6b.north west)+(-0.2,0.2)$) rectangle ($(v6b.south east)+(0.2,-0.2)$);
				
				\node[below=0.2cm of v3b,xshift=0.1cm] {\(\mc{S} = \{(1,1), (2,2), (1,3), (1,4), (2,5), (2,6)\}\)};
				\draw[dashed, thick, color=gray] 
				($(v1t.north west)+(-0.3,0.3)$) rectangle ($(v6b.south east)+(0.3,-0.3)$);
				\node[above=0.2cm of $(v6t.north)!0.0!(v6b.north)$,color=gray] {\(\mc{X}\)};
				\node[left=0.2cm of v1t, blue!60] {\(\mc{S}_1\)};
				\node[left=0.2cm of v1b, orange!60] {\(\mc{S}_2\)};
			\end{tikzpicture}
	\vspace{-0.2cm}
	\caption{(a) Disjoint selection of state $v$ from ground set \(\mc{V}\). (b) Disjoint selection from ground set \(\mc{X}=\mc{C}\times \mc{V}\) by duplicating each element \(v\in\mc{V}\) for each subsystem \(i\in\mc{C}\); selection under a partition matroid.
	}
	\label{fig:comparison_P1P2}
\end{figure}

\begin{mylem}\label{lemma:equivalenceP1andP2}
	Consider the observability-based partitioning problem \(\mb{P1}\) defined over the ground set \(\mc{V}\). By defining the ground set \(\mc{X} = \mc{C}\times \mc{V}\) and imposing the partition matroid constraint $ \{S \subseteq \mc{X} \mid |S\cap (\mc{C}\times \{v\})|\le 1,\ \forall\, v\in\mc{V}\},$
	\(\mb{P1}\) is equivalent to the submodular maximization formulation \(\mb{P2}\).
\end{mylem}

\begin{proof}
	We prove such equivalence by showing that the solution space of $\mb{P1}$ and $\mb{P2}$ is a bijective mapping. Suppose that $\mc{\pi}_{\kappa}$ is a feasible solution of $\mb{P1}$, then each $\mc{S}_i\subseteq\mc{V}$ represents a partition such that $$\cup_{i=1}^{\kappa} \mc{S}_{i}=\mc{V}, \; \mc{S}_{i} \cap \mc{S}_{j}=\emptyset, \;\; \forall \; i \neq j.$$ Then, for a ground set $\mc{X} = \mc{C}\times \mc{V}$, the partitions $\mc{S}_i$ can be written as their union according to $\mc{S}= \bigcup_{i=1}^{\kappa} \{(i,v) : v\in \mc{S}_i\} \subseteq \mc{X}$. The Cartesian product here ensures that each element $v \in \mc{V}$ appears $\kappa$ times for index $i \in \mc{C}$. This ensures that each state $v$ is assigned to subsystem $\mc{S}_i$ with no overlap, i.e., $\mc{S}_i\cap \mc{S}_j = \emptyset$; then the partition $\mc{S}$ has the constraint $\left|\mc{S} \cap \left( \mc{C} \times \{v\} \right) \right| \le 1, \quad \forall\, v \in \mc{V}$. That is, only one pair $(i,v) \in \mc{X}$ appears in $\mc{S}$. Thus $\mc{S}$ is equivalent to the feasible solution of $\mb{P2}$ due to constraint~\eqref{eq:P1-2-constraint}. Now, let $\pi^1_\kappa = \{\mc{S}^1_1,\dots,\mc{S}^1_\kappa\} \quad \text{and} \quad \pi^2_\kappa = \{\mc{S}^2_1,\dots,\mc{S}^2_\kappa\}$ be two feasible solutions of $\mb{P1}$ such that there exists at least one $i \in \mc{C}$ for which $\mc{S}^1_i \neq \mc{S}^2_i$, i.e., $\mc{S}^1_i \setminus \mc{S}^2_i \neq \emptyset$. Then for
	
	$$\mc{S}^1= \bigcup_{i=1}^{\kappa} \{(i,v) : v\in \mc{S}^1_i\}, \;\; \mc{S}^2= \bigcup_{i=1}^{\kappa} \{(i,v) : v\in \mc{S}^2_i\},$$ there exists at least some $v\in \mc{V}$ such that $v \in \mc{S}^1_i \setminus \mc{S}^2_i$. That is, there exists a pair $(i,v) \in \mc{S}$ that appears in $\mc{S}^1$ and not $\mc{S}^2$, i.e., $\mc{S}^1 \neq \mc{S}^2$. This establishes the injectivity of the feasible solution space map.
	
	Furthermore, suppose that $\mc{S}\subseteq \mc{X}$ is a feasible solution of $\mb{P2}$, then the constraint $\left|\mc{S} \cap (\mc{C}\times \{v\})\right| \leq 1, \; \forall \;  v \in \mc{V}$, ensures for each $(i,v) \in \mc{S}$, a state $v \in \mc{V}$ is assigned to $\mc{S}$. Now, for $i \in \mc{C}$ define $	\mc{S}_i = \{ v\in\mc{V} \,:\, (i,v)\in S\}$ then, the constraint of $\mb{P2}$ imposes that $\mc{S}_i \cap \mc{S}_j = \emptyset$ and since $v\in \mc{V}$ is allocated to at least one partition then $ \cup_{i=1}^{\kappa} \mc{S}_{i}=\mc{V}$. Thus $\{\mc{S}_1, \dots, \mc{S}_{\kappa}\}$ is equivalent to feasible solution of $\mb{P1}$ due to constraint~\eqref{eq:P1_contraint}.

	It now follows that for any feasible solution $\mc{S} \subseteq \mc{X}$ of $\mb{P2}$ and from the definition of $\mc{S}_i = \{v : (i,v) \in S\}$, a state $v \in \mc{V}$ is assigned to at most one partition $\mc{S}_i$. Thus, we obtain a feasible partition of $\mb{P1}$ $\{\mc{S}_1, \dots, \mc{S}_\kappa\}$ given that $ \cup_{i=1}^{\kappa} \mc{S}_{i}=\mc{V}$ (surjective property). Hence, the feasible space is a bijective mapping, meaning that the solutions map uniquely between the two problem formulations. It is straightforward to show that the objective function is equivalent under the mapping between $\mb{P1}$ and $\mb{P2}$. Let $\mc{S}_i = \{v: (i,v) \in S\}$, then~\eqref{eq:P1-2-obj} can be written as $f(S) = \sum_{i\in\mc{C}} f_i\big(\{v : (i,v) \in S\}\big) = \sum_{i\in\mc{C}} f_i(\mc{S}_i),$ which is equivalent to~\eqref{eq:P1_obj}. This concludes the proof.
\end{proof}

Following the proof of Lemma~\ref{lemma:equivalenceP1andP2}, in the next section we characterize subsystem observability and show that the observability-based measures defined by $f_i\left({\mc{S}_{i}}\right)$ are modular and submodular, and therefore maintain optimality guarantees when common greedy algorithms are implemented.
\begin{figure*}[t]
	\centering
	\vspace{+0.3cm}
	\captionsetup[subfloat]{labelformat=empty} 
	\makebox[\textwidth][c]{%
		\subfloat[\label{fig:sys_mapping}]{
		\begin{minipage}{0.48\textwidth}
			\centering
			{\small \textbf{(a) Subsystem Mapping}} \\[0.25cm]
			\begin{tikzpicture}[scale=0.9, transform shape, every node/.style={align=center, font=\scriptsize}]
				\node[draw=black, fill={rgb,255:red,215; green,238; blue,210}, rounded corners, minimum width=0.75cm, minimum height=2.2cm] (yS1) at (-1,1.15) {\scriptsize $\m{y}_{\mc{S}_1}$};
				\node[right=0.05cm and -0.05cm of yS1] {\scriptsize $=$};
				\node[draw=black, fill={rgb,255:red,251; green,234; blue,206}, rounded corners, minimum width=0.75cm, minimum height=2.2cm] (yS2) at (-1,-1.15) {\scriptsize $\m{y}_{\mc{S}_2}$};
				\node[right=0.05cm and -0.05cm of yS2] {\scriptsize $=$};
				
				\node[draw=black, fill={rgb,255:red,215; green,238; blue,210}, rounded corners, minimum width=4.3cm, minimum height=2.2cm] (S1) at (1.9,1.15) {};
				\node[draw=black, fill={rgb,255:red,251; green,234; blue,206}, rounded corners, minimum width=4.3cm, minimum height=2.2cm] (S2) at (1.9,-1.15) {};
				
				\node[draw, circle, fill=blue!10, minimum size=0.45cm] (v1) at (0.15, 1.85) {$v_1$};
				\node[draw, circle, fill=blue!10, minimum size=0.45cm] (v3) at (1.55, 1.15  ) {$v_3$};
				\node[draw, circle, fill=blue!10, minimum size=0.45cm] (v4) at (2.25, 0.45) {$v_4$};
				
				\node[draw, circle, fill=blue!10, minimum size=0.45cm] (v2) at (0.85, -0.45) {$v_2$};
				\node[draw, circle, fill=blue!10, minimum size=0.45cm] (v5) at (2.95, -1.1) {$v_5$};
				\node[draw, circle, fill=blue!10, minimum size=0.45cm] (v6) at (3.65, -1.85) {$v_6$};
				
				\node[draw, rectangle, rounded corners,minimum width=0.35cm, minimum height=4.5 cm, fill=gray!10] (xvec) at (4.55,0) {\scriptsize $\m{x}[k]$};
			\end{tikzpicture}
		\end{minipage}}
		\hfill
			\subfloat[\label{fig:gram_comp}]{
		\begin{minipage}{0.48\textwidth}
			\centering
			{\small \textbf{(b) Gramian Composition}} \\ 
			\begin{tikzpicture}[scale=1, transform shape, every node/.style={align=center, font=\scriptsize}]
				\node[draw, circle, minimum size=0.5cm, fill=gray!20] (sum1) at (0, 1.0) {$\sum$};
				\node at (0,0.4) {\scriptsize $i \in \mc{C}$};
				\node[draw, rectangle,rounded corners, fill={rgb,255:red,215; green,238; blue,210}, minimum width=1.8cm, minimum height=1.2cm, right=0.2cm of sum1] (W1) {$\m{W}_{\mc{S}_1}$};
				\node[draw, rectangle,rounded corners, fill={rgb,255:red,251; green,234; blue,206}, minimum width=1.8cm, minimum height=1.2cm, right=0.2cm of W1] (W2) {$\m{W}_{\mc{S}_2}$};
				\node[draw, rectangle,rounded corners, fill=gray!10, minimum width=1.0cm, minimum height=1.2cm, right=0.8cm of W2] (Wo) {$\m{W}_o$};
				
				\begin{scope}[on background layer]
					\node[
					draw=gray, 
					fill=gray!10, 
					rounded corners, 
					inner sep=2pt, 
					fit=(W1) (W2),
					label=above:{\scriptsize $\m{W}_{\pi_{\kappa}} = \sum_{i \in \mc{C}} \m{W}_{\mc{S}_i}$}
					] (SumW) {};
				\end{scope}
				\draw[-{Latex[length=2pt, width=2pt]}, thick,shorten <=1pt,shorten >=1pt] (SumW) -- (Wo) node[midway, above] {};
				
				\node[coordinate, below left=1.4cm and 1cm of sum1] (matrowstart) {};
				
				\node[draw, circle, minimum size=0.3cm, fill={rgb,255:red,215; green,238; blue,210}, below=0cm of matrowstart] (sumS1) {\tiny$\sum$};
				\node[below=0.01cm of sumS1] {\scriptsize $v \in \mc{S}_1$};
				
				\node[draw, minimum width=0.25cm, minimum height=1cm, fill=gray!10, right=0.1cm of sumS1] (Ak1) {$\m{A}^k$};
				\node[above=0.05cm of Ak1] {\tiny $n_x\hspace{-0.05cm} \times\hspace{-0.05cm} n_x$};
				
				\node[draw, minimum width=0.1cm, minimum height=1.6cm, fill={rgb,255:red,215; green,238; blue,210}, right=0.1cm of Ak1] (cvT) {$\m{c}_v^\top$};
				\node[below=0.05cm of cvT] {\tiny $n_x \hspace{-0.05cm}\times\hspace{-0.05cm} 1$};
				
				\node[draw, minimum width=1.0cm, minimum height=0.25cm, fill={rgb,255:red,215; green,238; blue,210}, right=0.1cm of cvT] (cv) {$\m{c}_v$};
				\node[below=0.45cm of cv] {\tiny $1\hspace{-0.05cm} \times\hspace{-0.05cm} n_x$};
				
				\node[draw, minimum width=0.25cm, minimum height=1cm, fill=gray!10, right=0.1cm of cv] (Ak2) {$\m{A}^k$};
				\node[above=0.05cm of Ak2] {\tiny $n_x \hspace{-0.05cm}\times\hspace{-0.05cm} n_x$};
				
				\node[coordinate, below right=1.4cm and 2.8cm of sum1] (matrowstart) {};
				
				\node[draw, circle, minimum size=0.3cm, fill={rgb,255:red,251; green,234; blue,206}, below=0cm of matrowstart] (sumS2) {\tiny $\sum$};
				\node[below=0.01cm of sumS2] {\scriptsize $v \in \mc{S}_2$};
				
				\node[draw, minimum width=0.25cm, minimum height=1cm, fill=gray!10, right=0.1cm of sumS2] (Ak12) {$\m{A}^k$};
				\node[above=0.05cm of Ak12] {\tiny $n_x\hspace{-0.05cm} \times\hspace{-0.05cm} n_x$};
				
				\node[draw, minimum width=0.1cm, minimum height=1.6cm, fill={rgb,255:red,251; green,234; blue,206}, right=0.1cm of Ak12] (cvT2) {$\m{c}_v^\top$};
				\node[below=0.05cm of cvT2] {\tiny $n_x \hspace{-0.05cm}\times\hspace{-0.05cm} 1$};
				
				\node[draw, minimum width=1.0cm, minimum height=0.25cm, fill={rgb,255:red,251; green,234; blue,206}, right=0.1cm of cvT2] (cv2) {$\m{c}_v$};
				\node[below=0.45cm of cv2] {\tiny $1\hspace{-0.05cm} \times\hspace{-0.05cm} n_x$};
				
				\node[draw, minimum width=0.25cm, minimum height=1cm, fill=gray!10, right=0.1cm of cv2] (Ak22) {$\m{A}^k$};
				\node[above=0.05cm of Ak22] {\tiny $n_x \hspace{-0.05cm}\times\hspace{-0.05cm} n_x$};
				
				\draw[decorate, decoration={brace, amplitude=1pt}, semithick]
				([yshift=0.4cm]Ak1.north west) -- ([yshift=0.4cm]Ak2.north east)
				node[midway, above=4pt] (brace1){};
				
				\draw[-] ($(brace1.center) + (0, -5pt)$) -- (W1.south);
				
				\draw[decorate, decoration={brace, amplitude=1pt}, semithick]
				([yshift=0.4cm]Ak12.north west) -- ([yshift=0.4cm]Ak22.north east)
				node[midway, above=4pt] (brace2) {};
				\draw[-] ($(brace2.center) + (0, -5pt)$) -- (W2.south);
			\end{tikzpicture}
		\end{minipage}
	}}
	\vspace{-0.2cm}
	\caption{(a) Mapping of state subsets $\mc{S}_i$ to outputs $\m{y}_{\mc{S}_i}[k]$; each $v_i$ corresponds to a row $\mb{c}_v$ of $\mb{C}$. (b) Composition of the full observability Gramian $\m{W}_o$ from subsystem Gramians $\m{W}_{\mc{S}_i}$ and the corresponding matrix structure in each summation term of the subsystem Gramians.}
	\label{fig:partition_observability}
	\vspace{-0.4cm}
\end{figure*}
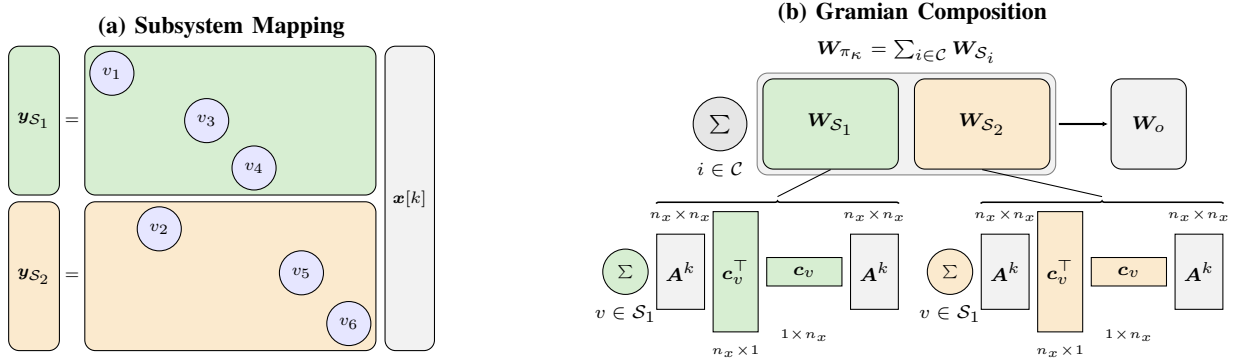
\section{Subsystem Observability under a Partition Matroid}\label{sec:main}
\subsection{Observability and its spectral measures}\label{subsec:systemdescription}
In this section, we present a brief overview of the notion of observability and its measures as a basis for quantifying subsystem observability under a partition matroid constraint. The system~\eqref{eq:model_DT} is said to be observable for an observation horizon ${N}$ if and only if the observability matrix given by 
\begin{equation}\label{eq:Obs_matrix}
	\rank{\m{O}}\hspace{-0.01cm} :=\hspace{-0.01cm}\rank{\left\{\m{C}, \m{C} \mA, \ldots, \m{C} \mA^{N-1}\right\}} = n_x,
\end{equation}
where $\m{O}(\m{x}[0]):=\m{O}\in \mathbb{R}^{Nn_{y} \times n_{x}}$ is full column rank. This observability rank condition is a qualitative metric that indicates the ability to infer states of a dynamical system by simply measuring its output. However, this is not indicative of how observable the dynamical system is. In contrast, observability metrics that quantify energy and volumetric related notions of the observability Gramian $\m{W}  \in \mathbb{R}^{n_x \times n_x}$, denoted as,
\begin{equation}\label{eq:LinObsGram}
	\m{W}:=\sum_{k=0}^{N-1}\left({\m{A}^{k}}\right)^{\top} \m{C}^{\top} \m{C}\m{A}^{k}= \m{O}^{\top} \m{O},
\end{equation}
provide a quantification of system observability. The matrix $\m{W}$ is non-singular if the system is observable over horizon $N$, otherwise it is not observable. That is, system~\eqref{eq:model_DT} is observable if and only if $\m{W}\succ 0$ for all $N>0$~\cite{Sato2024}. If the matrix $\m{A}$ of system~\eqref{eq:model_DT} is stable, the limit $\m{W}_o:=\lim _{N \rightarrow \infty} \m{W}$ satisfies the Lyapunov equation
\begin{equation}\label{eq:LyapEqLin}
	\m{A}^{\top}\m{W}_o\m{A} -\m{W}_o+ \m{C}^{\top}\m{C} = 0.
\end{equation}
Based on the Gramian~\eqref{eq:LinObsGram}, the following metrics can be considered to quantify system observability: $\rank(\m{W})$, $\trace(\m{W})$, and $\logdet{\m{W}}$---see~\cite{Pasqualetti2014a, Summers2016}. The $\trace(\m{W})$ captures the total average sensitivity of outputs to state perturbations, while $\log\!\det(\m{W})$ relates to the volume of the confidence ellipsoid for state-estimation error. The $\rank(\m{W})$ indicates the number of observable modes. Other metrics are used in the literature.
\subsection{Characterizing subsystem observability}\label{sec:subsysobs}
The partitioning problem $\mb{P2}$ clusters the measurable state-space into $\kappa$ subsystems belonging to set $\pi_{\kappa}$ that are interconnected. The objective is to identify subsystems that maximize a partition's observability measure. The linear measurement equation, $\m y[k] =  \m{C} \m{x}[{k}]$, can be decomposed to represent each of the subsystems belonging to $\pi_{\kappa}$. That is, the measurable state vector for $i \in \mc{C}$ can be written as $\m y_{\mc{S}_i}[k] =  \m{C}_{\mc{S}_i}\m{x}[{k}] \in \Rn{\kappa_i}$, where measurement mapping-function $\m{C}_{\mc{S}_i} \in \Rn{\kappa_i \times n_x}$ measures the states within a subsystem $\mc{S}_i$ for $i \in \mc{C}$, such that $|\mc{S}_i| = \kappa_i$; see Fig.~\ref{fig:sys_mapping} for an illustrative example. The set function $f_i({\mc{S}_i})$ for each of the subsystems quantifies an observability Gramian based measure. To that end, we define the observability Gramian of a partition $\mc{S}_i$, with $ i\in\mc{C}$, through the following proposition. 
\begin{myprs}\textit{(Subsystem observability Gramian)}\label{prs:obsGramparitions}
	Let the ground set $\mc{V} = \{ v\in\mbb{N}\,|\,0 < v \leq n_y\}$ represent the measurable state-space. The observability Gramian~\eqref{eq:LinObsGram} for a subsystem $\mc{S}_i \in \pi_{\kappa}$ can be written as
	\begin{align}\label{eq:subobsGram}
		\m{W}_{\mc{S}_i} := \sum_{v\in \mc{S}_i} 
		\left(\sum_{k=0}^{N-1}\left(\m{A}^{k}\right)^\top \hspace{-0.05cm}\m c_v^\top \m c_v \m{A}^{k}\right) \; \in \Rn{n_x\times n_x},
	\end{align}
	where $\m c_v\in\mbb{R}^{1\times n_x}$ is the $v$-th row of $\m C$ representing the mapping of state $v \in \mc{S}_i \subseteq \mc{V}$. Furthermore, the observability Gramian of a system partitioning $\pi_{\kappa}$ is the sum of Gramians for each subsystem and is given by
	\begin{equation}\label{eq:piGramian}
		\m{W}_{\pi_{\kappa}} := \sum_{i\in\mc{C}}\m{W}_{\mc{S}_i}  \; \in \Rn{n_x\times n_x}.
	\end{equation} 
\end{myprs}
\begin{proof}
	The proof follows from the definition of the observability Gramian~\eqref{eq:LinObsGram} of system~\eqref{eq:model_DT}. For the LTI measurement mapping matrix $\mC$ where $\m c_v \in\mbb{R}^{1\times n_x}$ for $v\in \mc{V}$, the decomposition of $\m{C}^\top \m{C}$ can be expressed using the outer product of $\m c_v$, that is $	\m{C}^\top \m{C} = \sum_{v=1}^{n_y} \m c_v^\top \m c_v$. Now under a disjoint partition of the measurable state-space $\pi_{\kappa}$ with \(\bigcup_{i\in\mc{C}} \mc{S}_i = \mc{V}\) and \(\mc{S}_i \cap \mc{S}_j = \emptyset\) for \(i\neq j\), we can write $\m{C}_{\mc{S}_i}^{\top}\m{C}_{\mc{S}_i}$ as $\sum_{v\in \mc{S}_i} \m{c}_v^\top \m{c}_v$ for subsystem $\mc{S}_i$ for all $v\in \mc{S}_i$. Now for any $\mc{S}_i \in \pi_{\kappa}$, we consider $\m{C}_{\mc{S}_i}^{\top}\m{C}_{\mc{S}_i}$ as the subsystem decomposition mapping of $\m{C}^\top \m{C}$ in the observability Gramian and by the linearity of a summation of disjoint sets, we obtain
	\begin{align*}\label{proof:prs1-1}
		\m{W}_{\mc{S}_i}\hspace{-0.1cm}=\hspace{-0.2cm} \sum_{k=0}^{N-1} \hspace{-0.1cm}\left(\m{A}^{k}\right)^\top\hspace{-0.1cm} \left( \sum_{v\in \mc{S}_i} \m{c}_v^\top \m{c}_v\hspace{-0.1cm} \right) \hspace{-0.1cm}\m{A}^{k}
		\hspace{-0.1cm}=\hspace{-0.2cm}\sum_{v\in \mc{S}_i}\hspace{-0.1cm} \left( \sum_{k=0}^{N-1}\hspace{-0.1cm} \left(\m{A}^{k}\right)^\top\hspace{-0.1cm} \m{c}_v^\top \m{c}_v \m{A}^{k}\hspace{-0.1cm} \right)\hspace{-0.1cm}.\hspace{-0.1cm}
	\end{align*}
	
	Now suppose that $\pi_{\kappa} = \{\mc{S}_1, \ldots, \mc{S}_{\kappa}\}$ is a partition of $\mc{V}$. Then, we can construct $\mc{C}_{\pi_{\kappa}} = \left\{c_v: v\in \mc{S}_i \right\}_{i\in\mc{C}}$. It follows that, $\m{C}_{\pi_{\kappa}}^\top \m{C}_{\pi_{\kappa}} = \sum_{v\in \mc{V}} \m{c}_v^\top \m{c}_v = \sum_{i\in \mc{C}} \; \sum_{v\in \mc{S}_i} \m{c}_v^\top \m{c}_v$. Then for an observability Gramian under the action of a partition measurement mapping function $\m{C}_{\pi_{\kappa}}^\top \m{C}_{\pi_{\kappa}}$, we obtain 
	\begin{align*}
		\m{W}_{\pi_{\kappa}} 
		&= \sum_{k=0}^{N-1} \left(\m{A}^{k}\right)^\top \left(\sum_{i\in \mc{C}} \sum_{v\in \mc{S}_i} \m{c}_v^\top \m{c}_v \right) \m{A}^{k},\\ 
		&=  \sum_{i\in \mc{C}}{\sum_{k=0}^{N-1} \left(\m{A}^{k}\right)^\top \left(\sum_{v\in \mc{S}_i} \m{c}_v^\top \m{c}_v\right) \m{A}^{k}}= \sum_{i\in \mc{C}} \m{W}_{\mc{S}_i}.
	\end{align*}
	This concludes the proof.
\end{proof}

The above proposition establishes the Gramian matrix that enables the quantification of subsystem and partitioned system observability, i.e., it establishes the extent to which a set of measurable states, $v \in \mc{S}_i \subseteq \mc{V}$, contribute to overall system observability; see~Fig.~\ref{fig:gram_comp}. Based on the subsystem observability Gramian~\eqref{eq:subobsGram}, the Gramian of partition $\pi_{\kappa}$ can be written as~\eqref{eq:piGramian} and is equivalent to the observability Gramian~\eqref{eq:LinObsGram}, according to the following corollary.

\begin{mycor}\textit{(Gramian equivalence)}\label{corr:obssystosubsys}
	Let $\pi_{\kappa}$ represent a system partition, then the observability Gramian~\eqref{eq:LinObsGram} of the LTI system~\eqref{eq:model_DT} can be written in terms of the subsystem observability Gramians~\eqref{eq:subobsGram}. That is, $\m{W}_{o}$ is equivalent to $	\m{W}_{\pi_{\kappa}}$.
\end{mycor}
\begin{proof}
	The proof follows from the proof of Proposition~\ref{prs:obsGramparitions}. For a disjoint partition $\pi_{\kappa}$ with $\bigcup_{i\in \mc{C}} \mc{S}_i = \mc{V}$, we have $\m{C}^{\top}\m{C} = \m{C}_{\pi_{\kappa}}^\top \m{C}_{\pi_{\kappa}} = \sum_{i\in \mc{C}} \; \sum_{v\in \mc{S}_i} \m{c}_v^\top \m{c}_v$. Then we obtain the observability Gramian~\eqref{eq:LinObsGram} by summing over all $ v\in\mc{V}$.					
	\begin{align*}
		\m{W}_o &= \sum_{v=1}^{n_y} \left( \sum_{k=0}^{N-1} \left(\m{A}^{k}\right)^\top \, \m c_v^\top \m c_v \, \m{A}^{k} \right),\\
		&= \sum_{i\in\mc{C}} \sum_{v\in \mc{S}_i} \left( \sum_{k=0}^{N-1} \left(\m{A}^{k}\right)^\top \, \m c_v^\top \m c_v \, \m{A}^{k} \right) = \sum_{i\in \mc{C}} \m{W}_{\mc{S}_i},
	\end{align*}
	and this concludes the proof.
\end{proof}

The above results allow us to define the observability of each subsystem within the system partition $\pi_{\kappa}$. Notice that for subsystems that are disjoint and cover the measurable state-space $\mc{V}$, the observability of the dynamical system can be expressed as the sum of the Gramians of each subsystem. This equivalence, illustrated in Corollary~\ref{corr:obssystosubsys}, allows us to evaluate the observability of each partition independently and then aggregate their contributions. This enables scalable and decentralized observability analysis of large-scale dynamical systems, in particular that of SP applications.
\subsection{Submodularity under system partitioning}\label{subsec:submodpartGram}
In this section, we show that the aforementioned observability measures for an LTI system, represented by a system partition $\pi_{\kappa}$, exhibit modular and submodular properties and thus provide theoretical guarantees under problem formulation $\mb{P2}$. The following lemma is essential in proving this result.
\begin{mylem}\textit{(\hspace{-0.012cm}\cite{Bilmes2022})}\label{lem:convex_comb}
	For set functions ${f}_{1}, {f}_{2}, \dots, {f}_{k}: 2^{\mc{V}}\rightarrow \mbb{R}_{\geq 0}$ that are submodular,
	any conic combination, that is, any weighted non-negative sum, ${f}(\mc{S})\hspace{-0.05cm}:=\hspace{-0.05cm}\sum_{i=1}^{k}w_{i}{f}_{i}$,
	is submodular, such that $w_{i}\geq 0 \;\forall \;i$.
\end{mylem}

Based on the above lemma, we can now establish the modularity of the subsystem observability Gramian as follows.
\begin{theorem}\textit{(Modularity of subsystem observability Gramian)}\label{theo:modularGrampartitions}
		Let \(\pi_{\kappa}=\{\mc{S}_i\}_{i\in\mc{C}}\) be a disjoint partition of the measurable-state set \(\mc{V}\). Then for each subsystem $\mc{S}_i$, the observability Gramian given by $\m{W}_{\mc{S}_i}$ is a modular set function for all $i \in \mc{C}$. Consequently, the overall observability Gramian for the partitioned system $\m{W}_{\pi_{\kappa}}$ is modular.
\end{theorem}
\begin{proof}
	For each measured state \(v\in\mc{V}\), we can define its contribution to the subsystem observability Gramian as $\m{W}_{\mc{S}_i}(v) \coloneqq \sum_{k=0}^{N-1} \left(\m{A}^{k}\right)^\top \, \m{c}_v^\top \m{c}_v\, \m{A}^{k}$.
	Then, by definition, the observability Gramian for a subsystem $i\in \mc{C}$ is given by $\m{W}_{\mc{S}_i} = \sum_{v\in \mc{S}_i} \m{W}_{\mc{S}_i}(v)$. Note that this can be written in the modular form as $\m{W}_{\mc{S}_i} = \m{W}_{\mc{S}_i} (\emptyset) + \sum_{v\in\mc{S}_i} \m{W}_{\mc{S}_i} (v).$ 
	As such, let \(\mc{A}\) and \(\mc{B}\) be two disjoint subsets of \(\mc{S}_i\). Then,
	$$\m{W}_{\mc{S}_i}(\mc{A}\cup\mc{B}) 
	=\sum_{v\in \mc{A}\cup \mc{B}} \m{W}_{\mc{S}_i}(v)
	=\sum_{v\in \mc{A}} \m{W}_{\mc{S}_i}(v) + \sum_{v\in \mc{B}} \m{W}_{\mc{S}_i}(v).$$
	This additivity property proves that Gramian \(\m{W}_{\mc{S}_i}\) is a modular set function. Now, from Proposition~\ref{prs:obsGramparitions} we have 
	$$
	\m{W}_{\pi_{\kappa}} = \sum_{i\in\mc{C}}\m{W}_{\mc{S}_i} =  \sum_{i=1}^{\kappa} \sum_{v\in \mc{S}_i} \m{W}_{\mc{S}_i}(v).
	$$
	Since this is a sum of modular functions, it follows from Lemma~\ref{lem:convex_comb} that $\m{W}_{\pi_{\kappa}}$ is modular.
\end{proof}
The above theorem shows that the observability Gramian for each subsystem and the overall partitioned system are modular set functions. Accordingly, Corollary~\ref{corr:SubPartObsMetrics} shows the submodular properties of the observability-based metrics under $\pi_{\kappa}$.
\begin{mycor}\textit{(Submodular properties under a system partition)}\label{corr:SubPartObsMetrics}
	Let the set functions $f_i:2^{\mc{V}} \to \mathbb{R}_{\geq0}$ represent an observability based-measure, and let $f(\mc{S}) = \sum_{i\in\mc{C}} f_i(\mc{S}_i)$. Then, the function $f\left(\mc{S}\right) = \sum_{i \in \mc{C}} \trace\left(\m{W}_{\mc{S}_{i}}\right)$ is a modular set function, while $f\left(\mc{S}\right) = \sum_{i \in \mc{C}} \rank\left({\m{W}_{\mc{S}_{i}}}\right)$ and $f\left(\mc{S}\right) = \sum_{i \in \mc{C}} \logdet{\m{W}_{\mc{S}_{i}}}$ are submodular monotone increasing set functions. 
\end{mycor}

\begin{proof}
	The proof follows from Theorem~\ref{theo:modularGrampartitions}. For each subsystem $\mc{S}_i \in \pi_{\kappa}$, the observability Gramian can be written as 	$\m{W}_{\mc{S}_i} = \sum_{v\in \mc{S}_i} \m{W}_{\mc{S}_i}(v)$. For disjoint subsystems $\mc{S}_i$ the trace is a linear function that depends on the measured states $v \in \mc{S}_i$. This implies that the trace is modular set function. Then, the sum of trace functions evaluated for $i\in \mc{C}$ we obtain
	\begin{align*}
	f(\mc{S}) &= \sum_{i \in \mc{C}} \trace(\m{W}_{\mc{S}_i}) 
		\hspace{-0.05cm}=\hspace{-0.05cm} \sum_{i \in \mc{C}} \trace\left(\sum_{v\in \mc{S}_i}\hspace{-0.1cm} \sum_{k=0}^{N-1}\hspace{-0.1cm} \left(\m{A}^{k}\right)^\top\hspace{-0.1cm} \m{c}_v^\top \m{c}_v \m{A}^{k}\hspace{-0.1cm} \right),\hspace{-0.1cm}\\
		&=\sum_{i \in \mc{C}} \sum_{v\in \mc{S}_i} \trace\left(\m{W}_{\mc{S}_i}(v)\right),
	\end{align*}	
	it is straightforward to show for any $\mc{A}_i,\,\mc{B}_i \subseteq\mc{S}_i$ and $\mc{A}_i \cup \{v_i\} = \mc{B}_i$ such that $i \in \mc{C}$, we obtain an equality
	\begin{equation}\label{eq:proof-cor2-1}
	\hspace{-0.2cm}\sum_{i=1}^{\kappa} \hspace{-0.05cm} \Big(  \hspace{-0.05cm} \trace(\mc{A}_i\cup\{v\}) \hspace{-0.05cm}- \hspace{-0.05cm}
	\trace(\mc{A}_i)\Big) \hspace{-0.05cm} =  \hspace{-0.05cm}\sum_{i=1}^{\kappa} \hspace{-0.05cm} \Big(\trace(\mc{B}_i\cup\{v\})  \hspace{-0.05cm}- \hspace{-0.05cm} \trace(\mc{B}_i) \hspace{-0.05cm}\Big),
	\end{equation}
	then, from Lemma~\ref{lem:convex_comb}, $f(\mc{S})\hspace{-0.1cm} =\hspace{-0.1cm}\sum_{i \in \mc{C}} \trace(\m{W}_{\mc{S}_i}) $ is a modular set function (see Definition~\ref{def:modular_submodular}). The proof of submodularity of the $\log\!\det$ and $\rank$ follow from Theorem~\ref{theo:modularGrampartitions} and Lemma~\ref{lem:convex_comb}, such that under the sum over all subsystems, $\log\!\det$ and $\rank$ exhibit monotone submodular properties. That is the inequality holds true in~\eqref{eq:proof-cor2-1} for $\mc{A}_i,\,\mc{B}_i \subseteq\mc{S}_i$ and $\mc{A}_i \cup \{v_i\} = \mc{B}_i$. We only provide the rest of the proof for the $\log\!\det$ metric, consider set function $ f_v: 2^{\mc{X} \setminus \{v\}} \rightarrow \mathbb{R}_{\geq0},
	\;\text{for each }v \in \mc{V},$ defined by
	\begin{align*}
		f_v(\mc{S})\hspace{-0.05cm}=\hspace{-0.05cm}\sum_{i \in \mc{C}} \logdet{\m{W}_{\mc{S}_i}(\mc{S}_i \cup\{v\})}\hspace{-0.05cm} - \hspace{-0.05cm}\sum_{i \in \mc{C}} \logdet{\m{W}_{\mc{S}_i}(\mc{S}_i)}.
	\end{align*}
	
	Now, let \(\mc{A}\) and \(\mc{B}\) be two sets such that $
	\mc{A} \;\subseteq\; \mc{B} \;\subseteq\; \mc{X}\setminus\{v\}$. Then we can define for each \(i \in \mc{C}\) the sets $\mc{A}_i := \mc{A} \cap \mc{S}_i, \; \mc{B}_i := \mc{B} \cap \mc{S}_i,$ such that
	\begin{align*}
		\overline{\m{W}}_{\mc{S}_i}(\epsilon)\hspace{-0.05cm}
		:= \hspace{-0.05cm}\m{W}_{\mc{S}_i}\!\bigl(\mc{A}_i\bigr)\hspace{-0.05cm}+\hspace{-0.05cm}
		\epsilon\,\Bigl( \m{W}_{\mc{S}_i}\!\bigl(\mc{B}_i\bigr) - \m{W}_{\mc{S}_i}\!\bigl(\mc{A}_i\bigr) \Bigr), \; 0 \le \epsilon \le 1.
	\end{align*}
	
	The set function $f_v(\mc{S})$ under the action of the above convex combination can be written as
	\begin{align*}
		\widehat{f}_v(\epsilon) \hspace{-0.05cm}:= \hspace{-0.05cm}\sum_{i \in \mc{C}}
		\Bigl[
		\logdet{\overline{\m{W}}_{\mc{S}_i}(\epsilon) \hspace{-0.05cm}+\hspace{-0.05cm} \m{W}_{\mc{S}_i}\bigl(\{v\}\hspace{-0.05cm}\bigr)}\hspace{-0.05cm} -\hspace{-0.05cm}
		\logdet{\overline{\m{W}}_{\mc{S}_i}(\epsilon)}\hspace{-0.05cm}
		\Bigr]\hspace{-0.05cm}.
	\end{align*}
	We can observe that $\widehat{f}_v(\epsilon)$ for $\epsilon =1$ evaluated at $\mc{S}_i \cup\{v\} \rightarrow \mc{B}_i$ and for $\epsilon =0$ evaluated at $\mc{S}_i \cup\{v\} \rightarrow \mc{A}_i$ we obtain the following relation
	\begin{align*}
		\widehat{f}_v(1) - \widehat{f}_v(0) = \int_0^1 \frac{d}{d\epsilon}\,
		\bigl[\widehat{f}_v(\epsilon)\bigr] \,d\epsilon.
	\end{align*} 
	
	We now show that the derivative of $\widehat{f}_v(\epsilon)$ written as 
	\(\tfrac{d}{d\epsilon}\widehat{f}_v(\epsilon) \) is negative
	for all \(\epsilon \in [0,1]\). Using the Jacobi formula
	$\tfrac{d}{d\epsilon} \logdet{ \m{W}_{\mc{S}_i}(\epsilon)} = 
	\trace\!\bigl[ \m{W}_{\mc{S}_i}(\epsilon)^{-1}\tfrac{d \m{W}_{\mc{S}_i}(\epsilon)}{d\epsilon}\bigr], $
	we obtain 
	\begin{align*}
		\frac{d}{d\epsilon}\,\widehat{f}_v(\epsilon) \hspace{-.05cm}= \hspace{-.05cm}&
		\sum_{i \in \mc{C}}
		\trace\Bigl[\hspace{-.05cm}
		\Bigl(\hspace{-.05cm}\bigl(\hspace{-.05cm}\overline{\m{W}}_{\mc{S}_i}(\epsilon) ) \hspace{-.05cm}+\hspace{-.05cm} \m{W}_{\mc{S}_i}\bigl(\{v\}\bigr)\bigr)^{-1} \hspace{-.1cm} -\hspace{-.1cm}
		\overline{\m{W}}_{\mc{S}_i}(\epsilon)^{-1} \hspace{-.05cm}\Bigr) \\
		&\quad \bigl( \m{W}_{\mc{S}_i}\bigl(\mc{B}_i\bigr) \;-\; \m{W}_{\mc{S}_i}\bigl(\mc{A}_i\bigr)
		\bigr)
		\Bigr],
	\end{align*}
	now since $ \overline{\m{W}}_{\mc{S}_i}(\epsilon) \;\preceq\; 
	\overline{\m{W}}_{\mc{S}_i}(\epsilon) \;+\; \m{W}_{\mc{S}_i}\bigl(\{v\}\bigr), $ it follows that
	\begin{align*}
		\Bigl(\,\overline{\m{W}}_{\mc{S}_i}(\epsilon) + \m{W}_{\mc{S}_i}\bigl(\{v\}\bigr)\Bigr)^{-1} -
		\overline{\m{W}}_{\mc{S}_i}(\epsilon)^{-1}\preceq 0, 
	\end{align*}
	is negative semidefinite. Therefore,
	\(\tfrac{d}{d\epsilon}\widehat{f}_v(\epsilon)\le 0\). Then, given that the derivative is non-positive on \([0,1]\), we have $	\widehat{f}_v(1) \le \widehat{f}_v(0)$ which implies that $	f_v(\mc{B}) \le f_v(\mc{A})$. This is equivalent to showing $f_v(\mc{A}\cup\{v\}) - f_v(\mc{A}) \ge
	f_v(\mc{B}\cup\{v\}) - f_v(\mc{B})$, i.e., the submodularity and monotone decreasing properties of $f_v$. By the additive property (Lemma~\ref{lem:convex_comb}) over all $v \in \mc{V}$, ${f}(\mc{S})$ is monotone increasing and submodular.
\end{proof}

Theorem~\ref{theo:modularGrampartitions} and Corollary~\ref{corr:SubPartObsMetrics} show that the observability-based measures for partitioning $\mb{P2}$ are submodular, and thus that $\mb{P2}$ can be solved using greedy algorithms with proven theoretical guarantees. The next section introduces the SP problem for a partitioned LTI system and shows that it can be cast as a submodular maximization problem under a partition matroid constraint, rather than a cardinality constraint (when the system is not partitioned). 
\section{Sensor Selection for Subsystem Dynamics}\label{sec:SNS}
We can now formulate the SP problem for the resulting subsystems defined by partitions $\mc{S}_i \in \pi^{*}_{\kappa}$. Given any LTI system and its $\kappa$ partitions obtained by solving $\mb{P2}$, the problem of allocating sensors across the subsystems is investigated.
\subsection{Subsystem sensor placement over a partition matroid}\label{subsec:SNS-1}
 Let the ground set $\mc{V}$ represent the measurable state-space. In other words, it is the set of possible available sensor location combinations. Let $\mc{R}\subseteq \mc{V}$ be a subset of sensor nodes where $v \in \mc{R}$ indicates an observed state. 
The number of sensors to be employed within the subsystems can be denoted as $\left|\mc{R}\right| = r$. The SP problem $(\mb{P3})$, over the subsystems, that is subject to a partition matroid can be written accordingly.
\begin{problem}\label{problem:P3}
	Given a system partitioning $\pi^*_{\kappa}=\{\mc{S}_i\}_{i\in\mc{C}}$ of measurable state-space $\mc{V}$, the problem of allocating $r$ sensors across $\kappa$ subsystems $\mc{S}_i \in \pi^*_{\kappa}$ while maximizing an overall observability-based measure $f(\mc{R}): 2^{\mc{V}} \rightarrow \mathbb{R}_{\geq0}$, with $f(\mc{R}) := f\left(\bigcup_{i\in\mc{C}}(\mc{R}\cap\mc{S}_i)\right)$, can be written as
\end{problem}
\begin{flalign}\label{eq:P3}
	\mb{P3}\hspace{-0.1cm}:&  \quad f^*_{\mc{R}}:=\max_{\mc{R}\subseteq\mc{V},\;\mc{R}\in\mc{I}_{p_{r_i}}}\;\;f(\mc{R})=f\left(\bigcup_{i\in\mc{C}}(\mc{R}\cap\mc{S}_i)\right)\hspace{-0.1cm},\hspace{-0.2cm}&
\end{flalign}
where the partition matroid constraint is given by $\mc{I}_{p_{r_i}}=\left\{\mc{R} \subseteq \mc{V}: \;  |\mc{R} \cap \mc{S}_i| \leq r_i, \forall \; i  \in \mc{C} \right\}$, such that $r_i$ is the number of sensors to be allocated to subsystem $\mc{S}_i$. The total number of sensors is therefore $r = \sum_{i \in \mc{C}} r_i$ and each $v \in \mc{R}$ corresponds to an allocated sensor.

\begin{table}[t]
	\fontsize{9}{9}\selectfont
	\centering 
	\caption{Notations for parameterized observability measures for a full and partitioned system.}
	\label{tab:obs_nomenclature}
	\renewcommand{\arraystretch}{1.3}
	\resizebox{\linewidth}{!}{
		\begin{tabular}{l|l|l}
			\toprule
			\textbf{Description} & \textbf{Full System} & \textbf{Partitioned System} \\
			\midrule
			Measurable state-space  & $\mc{V}$ & $\bigcup_{i \in \mc{C}} \mc{S}_i = \mc{V}$ \\
			Number of partitions & -- & $\kappa = |\mc{C}|$ \\
			System partition & -- & $\pi_{\kappa} = \{\mc{S}_1, \dots, \mc{S}_{\kappa}\}$ \\
			Subsystem& -- & A partition of $\pi_{\kappa}$, $\mc{S}_i \subseteq \mc{V}$ \\
			Sensor configuration set & $\mc{R} \subseteq \mc{V}$ & $\mc{R}_i = \mc{R} \cap \mc{S}_i \subseteq \mc{S}_i$ \\
			Overall system Gramian & $\m{W}_o$ & $\m{W}_{\pi_\kappa} = \sum_{i \in \mc{C}} \m{W}_{\mc{S}_i}$ \\
			Subsystem Gramian & -- & $\m{W}_{\mc{S}_i} = \sum_{v \in \mc{S}_i} \m{W}_{\mc{S}_i}(v)$ \\
			Parameterized Gramian& $\m{W}_o(\mc{R})$ & $\sum_{i \in \mc{C}} \m{W}_{\mc{S}_i}(\mc{R}_i)$ \\
			Observability-based measure & $f(\mc{R})$ & $f(\mc{R}) = f\left(\bigcup_i \mc{R}_i\right)$ \\
			$\mr{trace}$ metric & $\trace(\m{W}_o(\mc{R}))$ & $ \trace(\sum_{i \in \mc{C}}\m{W}_{\mc{S}_i}(\mc{R}_i))$ \\
			$\log\!\det$ metric & $\log\!\det(\m{W}_o(\mc{R}))$ & $\logdet{\sum_{i \in \mc{C}} \m{W}_{\mc{S}_i}(\mc{R}_i)}$ \\
			$\rank$ metric & $\rank(\m{W}_o(\mc{R}))$ & $\rank\left(\sum_{i \in \mc{C}} \m{W}_{\mc{S}_i}(\mc{R}_i)\right)$ \\
			\bottomrule
			\bottomrule
		\end{tabular}
	} 
\end{table}

The function $f({\mc{R}})= f\left(\bigcup_{i\in\mc{C}}(\mc{R}\cap\mc{S}_i)\right)$ is a submodular observability-based measure that quantifies overall system observability based on a set of sensor combinations $\mc{R} \subseteq \mc{V}$. Note that, if $f(\mc{R})$ is modular then it reduces to $f({\mc{R}}) = \sum_{i \in \mc{C}} f_i\left(\mc{R} \cap {\mc{S}_{i}}\right)$. The partition matroid ensures that the sensors are allocated across the partitions while maximizing an overall system observability measure $f(\mc{R})$. This means that there are no restrictions on having equal sensor ratios among the subsystems. The reason is due to the difference in objective functions between $\mb{P3}$ and $\mb{P2}$, where the observability-based measure of each individual partition or subsystem is maximized in $\mb{P2}$ rather than overall system observability $f\left(\bigcup_{i\in\mc{C}}(\mc{R}\cap\mc{S}_i)\right)$ in $\mb{P3}$. We summarize the notations for a full and partitioned system parameterized by sensor configuration set $\mc{R}$ in Tab.~\ref{tab:obs_nomenclature}.

We now show that the overall observability measure \(f(\mc{R})\), parameterized under a sensor set $\mc{R}$, retains its submodular set properties. Subsequently, we provide a theoretical optimality bound for the observability measures under a partition matroid.
\subsection{Submodular observability-based SP over a partition}\label{subsec:SNS-2}
To describe submodularity of the observability-based measure $f(\mc{R})$, evaluated under a system partition, we first introduce the following corollary that establishes the modularity of the subsystem observability Gramian under a specified sensor configuration. That is, the matrix $\m{W}_{\mc{S}_i}$ is modular in terms of sensor allocations within each subsystem $\mc{S}_i \in \pi_{\kappa}^{*}$. 

\begin{mycor}(Gramian modularity under sensor parameterization)\label{cor:ObsGramPartMod}
Let \(\pi_\kappa^{*}=\{\mc S_i\}_{i\in\mc C}\) be a partition of the measurable state-space \(\mc V\) obtained from solving $\mb{P2}$, and let \(\mc R\subseteq\mc V\) be any sensor configuration set. The modularity of the Gramian $\m{W}_{\mc{S}_i}$ for $i\in\mc{C}$ parameterized by a sensor selection set, denoted as $\mc{R}_i = \mc{R}\cap\mc{S}_i$, holds true for any $\kappa$ partitions.
\end{mycor}
\begin{proof}
	The proof follows from Theorem~\ref{theo:modularGrampartitions}. For a sensor configuration $\mc{R}$, we obtain a subsystem sensor configuration $\mc{R}_i = \mc{R}\cap\mc{S}_i$, such that for disjoint partitions $\mc{S}_i \in \pi_{\kappa}$, we have $\mc{R}_i \cap \mc{R}_j = \emptyset$ for all $i \neq j$. Then for $i\in\mc{C}$, we can write
	{\abovedisplayskip=0.1cm \belowdisplayskip=0.1cm
	\begin{align*}
		\m{W}_{\mc{S}_i}(\mc{R}_i) = \m{W}_{\mc{S}_i}(\mc{R}\cap\mc{S}_i)
		\;=\;\sum_{v\in \mc{R}\cap\mc{S}_i}\m{W}_{\mc{S}_i}(v),
	\end{align*}}
	where $\m{W}_{\mc{S}_i}(v)$ represents the contribution of measurable state $v\in \mc{V}$ to the $\m{W}_{\mc{S}_i} (\mc{R}_i)$ parameterized by sensor set $\mc{R}_i$. As such, for two disjoint subsets \(\mc{A}_i\) and \(\mc{B}_i\) of \(\mc{R}_i\), we obtain 
	$\m{W}_{\mc{S}_i}(\mc{A}_i\cup\mc{B}_i) 
	=\sum_{v \in \mc{A}_i\cup \mc{B}_i} \m{W}_{\mc{S}_i}(v)
	=\sum_{v\in \mc{A}_i} \m{W}_{\mc{S}_i}(v) + \sum_{v\in \mc{B}_i} \m{W}_{\mc{S}_i}(v)$. 
	This additivity property proves that Gramian \(\m{W}_{\mc{S}_i}\), under sensor parameterization $\mc{R}_i$, is a modular set function. Now, from Proposition~\ref{prs:obsGramparitions} we have 
	$
	\m{W}_{\pi_{\kappa}}(\mc{R})= \sum_{i\in\mc{C}}\m{W}_{\mc{S}_i} =  \sum_{i=1}^{\kappa} \sum_{v\in \mc{R}_i} \m{W}_{\mc{S}_i}(v).
	$
	Since this is a sum of modular functions, it follows from Lemma~\ref{lem:convex_comb} that $\m{W}_{\pi_{\kappa}}(\mc{R})$ is modular.
\end{proof}
Following Corollary~\ref{cor:ObsGramPartMod}, we now can establish the submodularity of the measures under a sensor parameterization.

\begin{mycor}\label{cor:SubObsMetrics}
	(Submodularity of parameterized observability measures under a partition) Given a set of sensors $\mc{R}_i = \mc{R}\cap\mc{S}_i$ within each subsystem, let $\m{W}_{\mc{S}_i}(\mc{R}_i)$ represent the contributions of sensors belonging to $\mc{R}_i$ for subsystem $\mc{S}_i$. Then, the parameterized set function $f(\mc{R}) = \trace\bigl(\cup_{i\in\mc{C}}\m{W}_{\mc{S}_i}( \mc{R}_i)\bigr)$ retains its modular set properties and $f(\mc{R}) = \log\!\det\bigl(\cup_{i\in\mc{C}}\m{W}_{\mc{S}_i}(\mc{R}_i)
	\bigr)$ and $f(\mc{R}) = \rank\bigl(\cup_{i\in\mc{C}} \m{W}_{\mc{S}_i}(\mc{R}_i) \bigr)$ remain submodular and monotone increasing under any sensor parameterization $\mc{R}$.
\end{mycor}
\begin{proof}
	From Corollary~\ref{cor:ObsGramPartMod} and for $i \in \mc{C}$ we can compute a sensor set $\mc{R}\subseteq\mc{V}$ contribution to the subsystem observability Gramian as $\m{W}_{\mc S_i}(\mc R_i)=\sum_{v\in\mc R_i}\m{W}_{\mc S_i}(v)$, where $\mc R_i=\mc R\cap\mc S_i$. Now considering the $\mr{trace}$ measure, the parameterized set function $f(\mc{R})$ can be written as
	{\abovedisplayskip=0.1cm \belowdisplayskip=0.1cm
	\begin{align*}
		f(\mc{R}) = \trace\Bigl( \bigcup_{i \in \mc{C}} \m{W}_{\mc{S}_i}(\mc{R} \cap \mc{S}_i) \Bigr)=\trace\Bigl( \sum_{i \in \mc{C}} \m{W}_{\mc{S}_i}(\mc{R} \cap \mc{S}_i) \Bigr),
	\end{align*}}
	\noindent where the second equality is due to the disjoint partition $\pi_{\kappa}$ such that for any sensor configuration $\mc{R}$, we have $\mc{R}_i \cap \mc{R}_j = \emptyset$ for all $i \in \mc{C}$. Now, since the $\mr{trace}$ is a linear operator and the parameterized subsystems are disjoint we can write
	\begin{align*}
	f(\mc{R}) =	\trace\Bigl( \sum_{i \in \mc{C}} \m{W}_{\mc{S}_i}(\mc{R} \cap \mc{S}_i) \Bigr)
	=  \sum_{i \in \mc{C}} \sum_{v \in \mc{R}\cap\mc{S}_i} \trace\Bigl(\m{W}_{\mc{S}_i}(v)\Bigr),
	\end{align*}
	now, similar to the proof of Corollary~\ref{cor:ObsGramPartMod}, it is straightforward to show that the equality holds true for~\eqref{eq:submodular_def} (Definition~\ref{def:modular_submodular}) for any $\mc{A}_i, \mc{B}_i \subseteq \mc{R}_i$ and $\mc{A}_i \cup \{v_i\} = \mc{B}_i$ for all $i\in\mc{C}$. Thus, $f(\mc{R}) = \trace\bigl(\cup_{i\in\mc{C}}\m{W}_{\mc{S}_i}( \mc{R}_i)\bigr)$ retains its modularity under sensor parameterization. 
	
	We follow a similar approach in Corollary~\ref{corr:SubPartObsMetrics} for the proof of the submodularity of the $\log\!\det$ of a system partition under sensor parameterization represented by $\mc{R}$. Again, we define an interpolation function to establish the submodularity of ${f}(\mc{R})$. Consider the set function $ f_v: 2^{\mc{V} \setminus \{v\}_i} \;\rightarrow\; \mathbb{R}_{\geq0},
	\;\text{for each }v \in \mc{V},$ defined by
	\begin{align*}
		f_v(\mc{R})\hspace{-0.05cm}=\hspace{-0.05cm} \logdet{\sum_{i \in \mc{C}}\m{W}_{\mc{S}_i}(\mc{R}_i \cup\{v\}_i)}\hspace{-0.1cm} - \hspace{-0.1cm}\logdet{\sum_{i \in \mc{C}}\m{W}_{\mc{S}_i}(\mc{R}_i)},
	\end{align*}
	where, for each \(i \in \mc{C}\), we denote \(\mc{R}_i := \mc{R} \cap \mc{S}_i\) and $\{v\}_i = v$ if $v\in \mc{S}_i$. Then, we can define a PSD matrix $\overline{\m{W}}(\epsilon)$ with $	0 \,\le\, \epsilon \,\le\, 1$ as follows
	\begin{align*}
		\overline{\m{W}}(\epsilon) =
		\sum_{i \in \mc{C}} \Bigl[
		\m{W}_{\mc{S}_i}(\mc{A}_i) + \epsilon 
		\Bigl( \m{W}_{\mc{S}_i}(\mc{B}_i) -  \m{W}_{\mc{S}_i}(\mc{A}_i) \Bigr)\Bigr],
	\end{align*}
	where $\overline{\m{W}}(\epsilon)$ interpolates over sets $\mc{A}_i := \mc{A} \cap \mc{S}_i,\; \text{and }
	\mc{B}_i := \mc{B} \cap \mc{S}_i,$ where $\mc{A}, \mc{B} \subseteq \mc{V}$ and $\mc{A} \subseteq \mc{B}$ such that $v \notin \mc{B}$. Now by considering
	\begin{align*}
		\widehat{f}_v(\epsilon) :=
		\log\!\det\!\Bigl(
		\overline{\m{W}}(\epsilon)+
		\sum_{i \in \mc{C}} \m{W}_{\mc{S}_i}\bigl(\{v\}_i\bigr)
		\Bigr) -
		\log\!\det\!\bigl(\overline{\m{W}}(\epsilon)\bigr).
	\end{align*}
	
	We can observe that by applying the derivative of $\widehat{f}_v(\epsilon)$ and showing in a similar methodology as in the proof of Corollary~\ref{corr:SubPartObsMetrics} that \(\tfrac{d}{d\epsilon}\widehat{f}_v(\epsilon)\le 0\). Then,  $\widehat{f}_v(1)
	\;=\;
	f(\mc{B}\cup\{v\}) \;-\; f(\mc{B}),
	\quad
	\widehat{f}_v(0)
	\;=\;
	f(\mc{A}\cup\{v\}) \;-\; f(\mc{A})$, such that  $\widehat{f}_v(1) \le \widehat{f}_v(0)$. Therefore establishing the submodularity of $f(R) =  \log\!\det\bigl(\cup_{i\in\mc{C}}\m{W}_{\mc{S}_i}(\mc{R}_i)
	\bigr)$. For brevity, we do not show the matrix derivative of $\log\!\det(\overline{\m{W}}(\epsilon))$ and solve for the derivative of $\widehat{f}_v(\epsilon)$. This follows a similar approach to the proof of Corollary~\ref{corr:SubPartObsMetrics}.
\end{proof}

Given that we are solving $\mb{P3}$ while maximizing an overall observability measure over the system partition rather than over individual subsystems, it is important to understand how the resulting sensor configuration for overall observability compares to that of subsystem observability. Note that, in the case where we are allocating sensors for local subsystem observability, we use $f(\mc{R}) = \sum_{i \in \mc{C}} f\left(\mc{R}\cap\mc{S}_i\right)$ when solving $\mb{P3}$. With that in mind, we show that there exists a bound for the different observability-based measures. The bound relates overall system observability to the sum of the observability metrics computed individually for each subsystem using the same sensor configuration. This result enables us to quantify the impact of decentralized sensor contributions on an overall observability-based objective. The following theorem establishes this bound on the resulting SP configuration within the subsystems as compared to overall system observability under the same sensor parameterization.
\begin{theorem}\label{theo:lowerbound}
	\textit{(Bound on parameterized subsystem observability measures)} Let $f: 2^{\mc{V}} \rightarrow \mathbb{R}_{\geq 0}$ be a $\mr{trace}$ or $\log\!\det$ set function quantifying the overall observability of parameterized subsystems $\mc{S}_i$. Then, for any set $\mc{R}$ with $|\mc{R}|=r$ sensors, the following holds true
	\begin{equation}\label{eq:lowerbound}
		f	\left(\bigcup_{i\in\mc{C}}(\mc{R}\cap\mc{S}_i)\right)\geq \sum_{i \in \mc{C}} f_i\left({\mc{R}\cap\mc{S}_{i}}\right),
	\end{equation}
	and the reverse inequality holds true for $\rank$, where $\pi_{\kappa}=\{\mc{S}_i\}_{i\in\mc{C}}$ is any partition of $\mc{V}$ and $f_i: 2^{\mc{V}} \rightarrow \mathbb{R}_{\geq 0}$ is the corresponding observability-based measure for subsystem $\mc{S}_i$.
\end{theorem}

\begin{proof}
	For each of the different measures we prove a lower bound for the observability of the LTI system accordingly. 
	\noindent \textit{(i)} The $\mr{trace}$ is a linear modular set function under sensor parameterization (Corollary~\ref{cor:ObsGramPartMod}), then the lower bound in~\eqref{eq:lowerbound} holds as an equality as follows.
	\begin{align*}
		\trace\Bigl( \bigcup_{i \in \mc{C}} \m{W}_{\mc{S}_i}(\mc{R} \cap \mc{S}_i)\Bigr) 
		=	\sum_{i \in \mc{C}}\trace\Bigl(\m{W}_{\mc{S}_i}(\mc{R} \cap \mc{S}_i) \Bigr).
	\end{align*}
	\noindent \textit{(ii)} The $\log\!\det$ is submodular and monotone increasing. For all subsystems Gramians \(\m{W}_{\mc{S}_i}(\mc{R}\cap\mc{S}_i)\succeq 0\), we employ the Minkowski determinant inequality~\cite[Theorem 7.8.21]{Horn1985}, which states that for any square matrices $\m{A}, \m{B} \succeq 0 $, the following inequality holds $\det(\m A+ \m B)^{\tfrac{1}{n}} \geq \det(\m A)^{\tfrac{1}{n}}+\det(\m B)^{\tfrac{1}{n}}$. By applying the logarithm to the Minkowski inequality of the sum of subsystem Gramians, the following holds true
	\begin{align*}
		\log\!\det\Bigl(\sum_{i\in\mc{C}} \m{W}_{\mc{S}_i}(\mc{R}_i) \Bigr)^{\tfrac{1}{n}}
		\geq \log\Bigl(\sum_{i\in\mc{C}} \det\bigl(\m{W}_{\mc{S}_i}(\mc{R}_i) \bigr)^{\tfrac{1}{n}}\Bigr).
	\end{align*}
	This inequality implies that the $\log\!\det$ of the sum of subsystem Gramians is greater than or equal to a function involving the log of the sum of determinants of $\m{W}_{\mc{S}_i}$. Now, given that the $\log\!\det$ is a concave function, and using Jensen's inequality~\cite[Lemma 1.1.2]{Niculescu2006} which states that $\sum_{i=1}^n \alpha_i f\left(\m{x}_i\right) \leq f\left(\sum_{i=1}^n \alpha_i \m{x}_i\right)$, where $\sum_{i=1}^n \alpha_i=1$. Then by defining $\alpha_i = \tfrac{r_i}{r}$, we can write that
	\begin{align*}
		\log&\left(\sum_{i \in \mc{C}} \tfrac{1}{\alpha_i}\det(\m{W}_{\mc{S}_i})\right)\geq \sum_{i \in \mc{C}}\tfrac{1}{\alpha_i}\log\!\det(\m{W}_{\mc{S}_i})\\ &\qquad= \log\bigl(\prod_{i \in \mc{C}} \det(\m{W}_{\mc{S}_i})^{\tfrac{1}{\alpha_i}}\bigr)= \sum_{i \in \mc{C}} \tfrac{1}{\alpha_i}\logdet{\m{W}_{\mc{S}_i}}.
	\end{align*}
	 Then, for $n=1$ and $\sum_{i\in \mc{C}}\alpha_i = \sum_{i\in \mc{C}}\tfrac{r_i}{r} = 1$, we obtain the lower bound as
	\begin{align*}
		\log\!\det\Bigl(\sum_{i\in\mc{C}} \m{W}_{\mc{S}_i}(\mc{R}_i) \Bigr) 
		\geq \sum_{i \in \mc{C}}\logdet{\m{W}_{\mc{S}_i}(\mc{R}_i)}.
	\end{align*}
	\noindent\textit{(iii)} For any partition with each \(\m W_{\mc S_i}(\mc R\cap\mc S_i)\succeq 0\) we obtain the following $\rank$ measure inequality for any sum of matrices.
	\begin{align*}
	\rank\Bigl(\sum_{i\in\mc{C}} \m{W}_{\mc{S}_i}(\mc{R}_i)\Bigr) \leq \sum_{i\in\mc{C}} \rank\Bigl(\m{W}_{\mc{S}_i}(\mc{R}_i)\Bigr).\label{eq:proof-theo2-rank}
	\end{align*}
	This completes the proof.
\end{proof}
\begin{myrem}\label{rmk:cardinalitytopartition}
	The set $\mc{R}$ representing the sensors allocated to the partitioned measurable state-space can yield sub-optimal results when considering the $\rank$ metric. By contrast, the optimality of the result for the $\mr{trace}$ and $\log\!\det$ measures is bounded by the sum of individual subsystem observability measures. The usage of specific observability metrics remains application-dependent but the theoretical results offer the described trade-offs between optimality and performance.
\end{myrem}

We note that a sensor $v \in \mc{R}$ is selected under a measurable state-space partition $\pi_{\kappa}$ with subsystems that are interconnected. This means that the contribution of this sensor within a subsystem impacts overall observability through the dynamical coupling of the states. As a result, and considering the above theorem, we can say that when solving the SP problem for the $\mr{trace}$ and $\log\!\det$ metrics over a system partition, we obtain an overall observability solution that is greater than or equal to the sum of the measures over each individual subsystem. This shows that considering a centralized SP formulation over an interconnected partitioned system results in performance that is bounded relative to the sum of subsystem observability measures. We show this in the case studies in Section~\ref{sec:simulation}.

In the next section, we present the framework and algorithm used to solve the observability-based partitioning and SP problems. This includes detailing the continuous relaxation of a submodular set maximization problem.

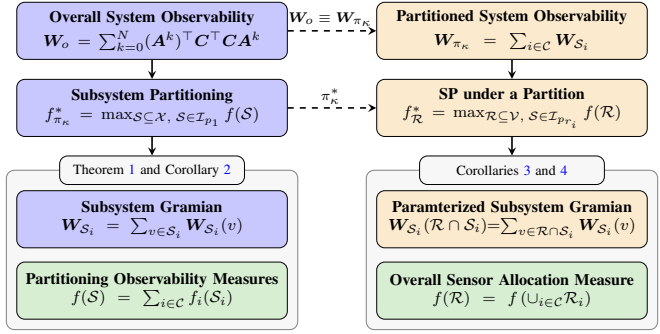
\begin{figure}[t]
	\vspace{+0.1cm}
	\centering
	\resizebox{1\linewidth}{!}{ 
		\begin{tikzpicture}[node distance=1.35cm]
			\node (fullsys) [colorbox1] {\textbf{Overall System Observability\vspace{+0.1cm}} \\ 
				$\m{W}_o = \sum_{k=0}^{N-1} (\m{A}^{k})^\top \m{C}^\top \m{C} \m{A}^{k}$}; 
			
			\node (partitioning) [colorbox1, below of=fullsys] {\textbf{Subsystem Partitioning} \\ $f^{*}_{\pi_{\kappa}}= \max_{\mc{S}\subseteq\mc{X},\; \mc{S} \in \mc{I}_{p_1}}  f(\mc{S})$};
			
			\node (obsmetric) [colorbox1, below of=partitioning, yshift=-0.6cm] {\textbf{Subsystem Gramian} \\ $\m{W}_{\mc{S}_i} = \sum_{v\in \mc{S}_i} \m{W}_{\mc{S}_i}(v)$};
			
			\node (submod) [colorbox3, below of=obsmetric,yshift=0.1cm] {\textbf{Partitioning Observability Measures} \\ 
				$f(\mc{S}) = \sum_{i \in \mc{C}} f_i(\mc{S}_i)$};
			
			\node (partsys) [colorbox2, right=1.55cm of fullsys] {\textbf{Partitioned System Observability\vspace{+0.1cm}} \\ 
				$\m{W}_{\pi_{\kappa}} = \sum_{i\in\mc{C}} \m{W}_{\mc{S}_i}$}; 
			
			\node (sensorsel) [colorbox2, right=1.55cm of partitioning] {\textbf{SP under a Partition} \\ 
				$f^*_{\mc{R}}= \max_{\mc{R}\subseteq\mc{V},\; \mc{R} \in \mc{I}_{p_{r_i}}}  f(\mc{R})$};
			
			\node (matroid) [colorbox2, below of=sensorsel, yshift=-0.6cm] {\textbf{Parameterized Subsystem Gramian} \\ 
				$\m{W}_{\mc{S}_i}(\mc{R}\cap\mc{S}_i) \hspace{-0.1cm}=\hspace{-0.1cm}\sum_{v\in \mc{R}\cap\mc{S}_i}\m{W}_{\mc{S}_i}(v)$};
			
			\node (finalsol) [colorbox3, below of=matroid, yshift=0.1cm] {\textbf{Overall Sensor Allocation Measure} \\ 
				$f(\mc{R}) =f\left(\cup_{i\in\mc{C}} \mc{R}_{i}\right)$};
			
			\draw [arrow] (fullsys) -- (partitioning);

			\draw [dashedarrow] (fullsys.east) -- ++(1cm,0) |- (partsys.west);
			\draw [dashedarrow] (partitioning.east) -- ++(1cm,0) |- (sensorsel.west);
			
			\draw [arrow] (partsys) -- (sensorsel);
			
			\begin{scope}[on background layer]
				\coordinate (extra1) at ($(obsmetric.north) + (0,5pt)$);
				\node[
				draw=gray, 
				fill={rgb,255:red,247; green,247; blue,247}, 
				rounded corners, 
				inner sep=5pt, 
				fit= (obsmetric) (submod) (extra1),
				] (block1) {};
			\end{scope}
			
			\node (theo1) [colorbox4, anchor =north] at ([xshift=0cm, yshift=+0.25cm]block1.north) 			{Theorem~\ref{theo:modularGrampartitions} and Corollary~\ref{corr:SubPartObsMetrics}};
			
			\begin{scope}[on background layer]
				\coordinate (extra2) at ($(matroid.north) + (0,5pt)$);
				\node[
				draw=gray, 
				fill={rgb,255:red,247; green,247; blue,247}, 
				rounded corners, 
				inner sep=5pt, 
				fit= (matroid) (finalsol) (extra2),
				] (block2) {};
			\end{scope}

			\node (theo2) [colorbox4, anchor =north] at ([xshift=0cm, yshift=+0.25cm]block2.north) 	{Corollaries~\ref{cor:ObsGramPartMod} and~\ref{cor:SubObsMetrics}};
			
			\draw [arrow] (partitioning) -- (theo1);
			\draw [arrow] (sensorsel) -- (theo2);
			
			\node (pistar) 
			at ($($(fullsys.east)!0.45!(partsys.west)$)+(0.05cm,0.2cm)$) 
			{\scriptsize $\m{W}_o \equiv \m{W}_{\pi_{\kappa}}$};
			\node (pistar) 
			at ($($(partitioning.east)!0.45!(sensorsel.west)$)+(0.05cm,0.2cm)$) 
			{\scriptsize $\pi_{\kappa}^{*}$};
			
		\end{tikzpicture}
	} 
	\vspace{-0.2cm}
	\caption{The framework for submodular partitioning and sensor placement in LTI systems. 
		The left column partitions the measurable state-space by solving $\mb{P2}$ while quantifying observability based on $\kappa$ subsystems.
		The right column represents the resulting partitioned system, where $\mb{P3}$ is solved under a partition matroid constraint. Gray boxes show the modular and submodular properties of the subsystem observability Gramians, as established in Theorem~\ref{theo:modularGrampartitions} and Corollaries~\ref{corr:SubPartObsMetrics}--\ref{cor:SubObsMetrics}.}\label{fig:submodular_P1-P2}
\end{figure}
\section{Partitioning and observability via the Continuous Multilinear Extension}\label{sec:multilinear}
The framework of observability-based partitioning and SP for an LTI system is illustrated in Fig.~\ref{fig:submodular_P1-P2}. As a result of posing the problems under a matroid constraint, and due to the submodular properties of the observability-based measures illustrated in Theorem~\ref{theo:modularGrampartitions} and Corollaries~\ref{corr:SubPartObsMetrics}--\ref{cor:SubObsMetrics}, we exploit greedy algorithms to solve observability-based system partitioning and SP over the interconnected subsystems.

For a matroid constraint, solving the submodular maximization problem continuously offers a better optimality guarantee. This can be done by applying continuous extensions to a submodular function; that is, extending set function $f$ to a function $F: [0,1]^{n} \rightarrow \mathbb{R}_{\geq 0}$, which agrees with $f$ on the hypercube vertices~\cite{Bai2018}. An extension to submodular functions under maximization, such as the multilinear extension~\cite{Vondrak2008a}, is equivalent to the approximate concave closure of the objective set function. The application of the multilinear relaxation to a utility-maximization problem of multi-agent systems under a partition matroid is studied in~\cite{Rezazadeh2021}. In this work, we apply such an extension to solve the partitioning problem $\mb{P2}$ and the SP problem $\mb{P3}$ using the continuous greedy algorithm. We note that distributed algorithms, such as the one introduced in~\cite{Rezazadeh2021}, can be utilized for both problem formulations in this manuscript; however, for brevity, we do not attempt to solve the partitioning and SP problems using different algorithms. Furthermore, for succinctness in exposition, we write the extension for the partitioning problem $\mb{P2}$; an equivalent formulation can be written for $\mb{P3}$ by changing the ground set and the partition matroid constraint.

For a monotone submodular function $f: 2^{\mc{X}} \rightarrow \mbb{R}_{\geq 0}$, we consider its multilinear extension $F: [0,1]^n \rightarrow \mbb{R}_{\geq 0}$ as 
\begin{equation}\label{eq:multilinext}
	F(\mathbf{x}) = \sum_{\mc{S} \subseteq \mc{X}} f(\mc{S}) \prod_{s \in \mc{S}}[\mathbf{x}]_s \prod_{s \notin \mc{S}}(1 - [\mathbf{x}]_s),
\end{equation}
where $[\mathbf{x}]_s$ denotes the inclusion probability of element $s \in \mc{X}$ and $n = |\mc{X}|$. Let the set $\mc{S}_{\mathbf{x}}$ be any subset for $\mb{x} \in[0,1]^n$ where each element $s\in\mc{X}$ is included independently with probability $[\mathbf{x}]_s$. Then, the multilinear extension can be equivalently expressed by randomly sampling sets $\mc{S}$ to probabilities in $[\mathbf{x}]_s$ as the expected value, $F(\mathbf{x}) = \mathbb{E}\left[f(\mc{S}_{\mathbf{x}})\right]$, where $\mathbb{E}[\cdot]$ denotes the expected value of $f\left(\mc{S}_{\mb{x}}\right)$. The submodular set maximization problem under the application of the multilinear extension can be written as
\begin{equation}\label{eq:sns_submod_extended}
	F^*_{\mc{S}} :=\max_{\mc{S}\subseteq\mc{X},\; \mc{S} \in \mc{I}_{p_1}} \;\;\;\; F(\mb{x}).
\end{equation}

The first partial derivative of the multilinear extension $F(\mb{x})$ can be written as
\begin{align}
	\frac{\partial F(\mb{x})}{\partial[\mb{x}]_s} &= \mathbb{E}\left[f(\mc{S}_{\mb{x}} \cup \{s\}) - f(\mc{S}_{\mb{x}} \setminus \{s\})\right]. \label{eq:partialF}
\end{align}

Readers are referred to~\cite{Vondrak2008a} for the derivation of the above partial derivative. Note that $F$ is concave in every direction when $f$ is polymatroid.
This means that $\tfrac{\partial F(\mb{x})}{\partial[\mb{x}]_s} \geq 0$ if and only if $f$ is monotone non-decreasing, and the second partial derivative, $\tfrac{\partial^2 F(\mb{x})}{\partial[\mb{x}]_a \partial[\mb{x}]_b}$, for distinct $a, b \in \mc{X}$, is $\leq 0$ if and only if $f$ is submodular.

To solve this continuous relaxation, we apply the continuous greedy Algorithm~\ref{alg:continuous_greedy}. The algorithm follows the direction of maximum expected marginal gain under a partition matroid constraint. The algorithm defines a path $\mb{x}:[0,1] \rightarrow \mc{S}_{\mb{x}}$, where $\mb{x}(0)=0$ and $\mb{x}(1)$ is the output of the algorithm. The gradient of $\mb{x}$ is chosen greedily in $\mc{X}$ in order to maximize $F$. This results in maximizing $\frac{\mathrm{d}}{\mathrm{d} l} \mb{x}(l) =\argmax_{\mb{x}\in \mc{S}_{\mb{x}}} \frac{\partial F}{\partial[\mb{x}]_s}(\mb{x}(l))$, where $l \in [0,1]$. Due to the equality in~\eqref{eq:partialF}, this is equivalent to maximizing the expected marginal gain as follows
\begin{align}\label{eq:expected_value}
	\underset{\mc{S}\in \mc{I}}{\mr{argmax}} \sum_{s \in \mc{S}} w_s \sim \mathbb{E}\left[f\left(\mc{S}_{\mb{x}} \cup\{s\}\right)-f\left(\mc{S}_{\mb{x}} \backslash\{s\}\right)\right],
\end{align}
Once a fractional solution $\mb{x}^*$ is obtained, a \emph{pipage rounding}~\cite{Ageev2004} algorithm is used to recover a discrete feasible solution. It is shown in~\cite{Calinescu2011}, that pipage rounding can be used to round the solution while maintaining its feasibility; interested readers are referred to~\cite[Section 3.2]{Calinescu2011}.
\begin{myrem}
	In $\mb{P2}$, the underlying matroid constraint is a simple partition matroid, $\kappa_i=1$. Thus, the continuous greedy algorithm can be simplified by performing independent randomized rounding for each state after obtaining the fractional solution $\mb{x}$. This rounding preserves the $(1-1/e)$ guarantee without requiring pipage rounding~\cite{Calinescu2011}.
\end{myrem}
\begin{algorithm}[t]
	\caption{Continuous greedy algorithm~\cite{Calinescu2011} for solving $\tikzmarknode{P1-0}{\mb{P2}}$ and $\tikzmarknode{P2-0}{\mb{P3}}$.}\label{alg:continuous_greedy}
	\DontPrintSemicolon
	\textbf{Input:} Ground set $\tikzmarknode{P1-1}{\mc{X}}$ , $\tikzmarknode{P2-1}{\mc{V}}$ , function $\tikzmarknode{P1-2}{f(\mc{S})}$ , $\tikzmarknode{P2-2}{f(\mc{R})}$ , matroid $\tikzmarknode{P1-3}{\mc{I}_{p_1}}$ , $\tikzmarknode{P2-3}{\mc{I}_{p_{r_i}}}$ , partitions $\kappa$, sensors $r$\;\vspace{0.1cm}
	\textbf{Initialize:} $\mb{x} \leftarrow \mathbf{0}$, $\tikzmarknode{P1-4}{\mc{S} \leftarrow \mb{0}}$ , $\tikzmarknode{P2-4}{\mc{R} \leftarrow \mb{0}}$\;
	
	\For{$t \leftarrow 1$ \KwTo $T$}{
		Estimate $\frac{\partial F}{\partial [\mb{x}]_s}$ for either $\tikzmarknode{P1-8}{s \in \mc{X}}$ or $\tikzmarknode{P2-8}{r \in \mc{V}}$\;\vspace{0.1cm}
		\tikzmarknode{P1-5}{Select $\mc{S}^* = \arg\max_{\mc{S} \in \mc{I}_{p_1}} \sum_{s \in \mc{S}} \frac{\partial F}{\partial [\mb{x}]_s}\vspace{0.2cm}$ using~\eqref{eq:expected_value},}\vspace*{0.2cm}
		\tikzmarknode{P2-5}{or $\mc{R}^* = \arg\max_{\mc{R} \in \mc{I}_{p_{r_i}}} \sum_{r \in \mc{R}} \frac{\partial F}{\partial [\mb{x}]_r}$ using~\eqref{eq:expected_value}}\;\vspace{0.2cm}
		\tikzmarknode{P1-6}{Update $\mb{x} \leftarrow \mb{x} + \frac{1}{\kappa} \cdot \mathbf{1}_{\mc{S}^*}$,}\hspace{0.1cm}
		\tikzmarknode{P2-6}{or $ \mb{x} \leftarrow \mb{x} + \frac{1}{\kappa} \cdot \mathbf{1}_{\mc{R}^*}$}\;
	}
	Apply random sampling to obtain $\tikzmarknode{P1-7}{\mc{S}^* \subseteq \mc{X}}$ , or pipage rounding to obtain $\tikzmarknode{P2-7}{\mc{R}^* \subseteq \mc{V}}$\;
\end{algorithm}

\tikz[overlay, remember picture]{
	\node[highlightP1] at (P1-0.base) {\(\mb{P2}\)};}
\tikz[overlay, remember picture]{
	\node[highlightP1] at (P1-1.base) {\(\mc{X}\)};}
\tikz[overlay, remember picture]{
	\node[highlightP1] at (P1-2.base) {\(f(\mc{S})\)};}
\tikz[overlay, remember picture]{
	\node[highlightP1] at (P1-3.base) {\(\mc{I}_{p_1}\)};}
\tikz[overlay, remember picture]{
	\node[highlightP1] at (P1-4.base) {\(\mc{S} \leftarrow \mb{0}\)};}
\tikz[overlay, remember picture]{
	\node[highlightP1] at (P1-5.base) {Select $\mc{S}^* = \arg\max_{\mc{S} \in \mc{I}_{p_1}} \sum_{s \in \mc{S}} \frac{\partial F}{\partial [\mb{x}]_s}$ using~\eqref{eq:expected_value},\vspace{0.1cm}};}
\tikz[overlay, remember picture]{
	\node[highlightP1] at (P1-6.base) {Update $\mb{x} \leftarrow \mb{x} + \frac{1}{\kappa} \cdot \mathbf{1}_{\mc{S}^*}$,}; }
\tikz[overlay, remember picture]{
	\node[highlightP1] at (P1-7.base) {\(\mc{S}^* \subseteq \mc{X}\)};}
\tikz[overlay, remember picture]{
	\node[highlightP1] at (P1-8.base) {\(s \in \mc{X}\)};}

\tikz[overlay, remember picture]{
	\node[highlightP2] at (P2-0.base) {\(\mb{P3}\)};}
\tikz[overlay, remember picture]{
	\node[highlightP2] at (P2-1.base) {\(\mc{V}\)};}
\tikz[overlay, remember picture]{
	\node[highlightP2] at (P2-2.base) {\(f(\mc{R})\)};}
\tikz[overlay, remember picture]{
	\node[highlightP2] at (P2-3.base) {\(\mc{I}_{p_{r_i}}\)};}
\tikz[overlay, remember picture]{
	\node[highlightP2] at (P2-4.base) {\(\mc{R} \leftarrow \mb{0}\)};}
\tikz[overlay, remember picture]{
	\node[highlightP2] at (P2-5.base) {or $\mc{R}^* = \arg\max_{\mc{R} \in \mc{I}_{p_{r_i}}} \sum_{r \in \mc{R}} \frac{\partial F}{\partial [\mb{x}]_r}$ using~\eqref{eq:expected_value}\vspace{0.1cm}};}
\tikz[overlay, remember picture]{
	\node[highlightP2] at (P2-6.base) {or $\mb{x} \leftarrow \mb{x} + \frac{1}{\kappa} \cdot \mathbf{1}_{\mc{R}^*}$};}
\tikz[overlay, remember picture]{
	\node[highlightP2] at (P2-7.base) {\(\mc{R}^* \subseteq \mc{V}\)};}
\tikz[overlay, remember picture]{
	\node[highlightP2] at (P2-8.base) {\(r\in \mc{V}\)};}

The partitioning problem $\mb{P2}$ can be efficiently solved while obtaining the following optimality guarantee.
\begin{theorem}(\hspace{-0.012cm}\cite{Calinescu2011})~\label{theo:guarenteeMulti}
	Let $f: 2^\mc{X} \rightarrow \mathbb{R}_{\geq0}$ be a polymatroid function and $F:[0,1]^n \rightarrow \mathbb{R}_{\geq0}$ be its multilinear extension.
	Then, the following performance bound holds true
	\begin{align*}
		F^*-f(\emptyset) \geq \left(1-\frac{1}{e}\right)\left(f^*-f(\emptyset)\right), \quad \text{with}\; f(\emptyset) =0,
	\end{align*}
	where $f^*$ and $F^*$ are the optimal values of the set maximization problem and its continuous relaxation.
\end{theorem}

Note that the same optimality guarantee holds true for the SP problem $\mb{P3}$. The continuous greedy algorithm presented in Algorithm~\ref{alg:continuous_greedy} can be used to solve both $\mb{P2}$ and $\mb{P3}$. The differences are depicted by the colors highlighting the equations for each of the problems. The computational complexity of the continuous greedy algorithm in general is $\mc{O}(|X|^7)$~\cite{Vondrak2008a}, while for a simple partition matroid as in $\mb{P2}$, it is $\mc{O}(|X|^2)$~\cite{Calinescu2011}, where $|X| = \kappa \times n_y$. There exists a plethora of studies that contribute to finding optimal and efficient methods to solve and formulate the continuous maximization problem~\eqref{eq:sns_submod_extended}. Methods include \textit{(i)} the probabilistic formulation utilized herein~\cite{Calinescu2011}, \textit{(ii)} exact multilinear extensions for specific objective functions~\cite{Iyer2014}, and approximate extensions by Taylor series approximations~\cite{Ozcan2021}; both do not require random sampling. Indeed, investigating the applicability of different approaches to solve the multilinear extension is warranted. However, the objective of this manuscript is to formulate the partitioning problem and its respective SP formulation while investigating their submodular properties.
\section{Numerical Simulation}\label{sec:simulation}
To demonstrate and validate the measurable state partitioning problem $\mb{P2}$ and the subsequent SP problem $\mb{P3}$, we investigate the following research questions
\begin{itemize}
	\item$(\mr{Q}1)$ How does solving $\mb{P3}$ over a partitioned system compare to solving the SP problem over the unpartitioned system?
	\item$(\mr{Q}2)$ Does the inequality established in Theorem~\ref{theo:lowerbound} hold for $\log\!\det$ observability measure parameterized by sensor configuration $\mc{R}$? 
	\item$(\mr{Q}3)$ How does the choice of the number of partitions $\kappa$ affect the optimality of the measurable state-space partitioning in $\mb{P2}$? Does the framework enable computational scalability, and what are the trade-offs between partition number and SP performance?
	\item $(\mr{Q}4)$ How does the observability of the subsystems obtained from $\mb{P2}$ compare to that of partitions from the spectral clustering method? Is the state-estimation performance maintained under the proposed partitioning framework?
\end{itemize} 
To address the above questions, we consider a nonlinear natural gas combustion reaction network~\cite{smirnov2014modeling} written as
\begin{equation}\label{eq:combustion_network}
	\dot{\m {x}}(t)=\Theta \boldsymbol{\psi}\left(\m {x}(t)\right),
\end{equation}
where $\m{\psi}\left(\m {x}\right)=[\psi_{1}\left(\m {x} \right),\psi_{2}\left(\m {x} \right),\ldots, \psi_{N_{r}}\left(\m {x} \right)]^{\top}$, and $\psi_{j}$, $j\in\{1,\;2,\;\ldots\;, \; N_{r}\}$ are the polynomial functions of concentrations. The state vector $\m {x}=[x^{1},\; x^{2},\;\cdots\;,\; x^{n_{x}}]$, represents the concentrations of $n_x$ chemical species. The constant matrix $\Theta=[w_{ji}-q_{ji} ]\in \mathbb{R}^{n_{x}\times N_{r}}$, where $q_{ji}$ and $w_{ji}$ are the stoichiometric coefficients. Let $N_{r}$ denote the number of chemical reactions, then the list of reactions can be written as $\sum_{i=1}^{n_{x}}q_{ji} \mc{R}_{i} \rightleftarrows  \sum_{i=1}^{n_{x}} w_{ji}\mc{R}_{i}, \; j \in \{1,\;2\;,\;\cdots\;, \; N_{r}\}$, where $\mc{R}_{i}$, $i\in \{1,\;2,\;\cdots\;,\; n_{x}\}$ are the chemical species. The reaction rates are computed using Cantera~\cite{Cantera} at a temperature of $1000 \, \mr{K}$ and under atmospheric pressure. 

The two mechanisms or combustion networks considered are: $\mr{(N1)}$ an $\mr{H}_2\mr{O}_2$ network that has $N_r = 27$ reactions and $n_x=10$ chemical species and $\mr{(N2)}$ a $\mr{GRI30}$ network that has $N_r = 325$ reactions and $n_x=53$ chemical species. The network models are described using Cantera type files ``h2o2.yaml" and ``gri30.yaml". The combustion reaction is solved using the Cantera package. The submodular maximization problems are implemented in Python and solved on a MacBook Pro with an Apple M1 Pro chip, a 10-core CPU, and 16 GB RAM. We linearize the nonlinear combustion network~\eqref{eq:combustion_network} around reference steady-state conditions $\m{x}_0 = \m{x}_0^{\mr{act}}$ obtained from the mechanism's nominal equilibrium. The discretization constant using forward Euler is $\Delta\mr{t}=1\cdot 10^{-12}$ and an observation window of $N = 1000$ is chosen. We perturb the initial conditions using a conservation-respecting transformation that lies in the null space of an elemental composition matrix $\m{E}$. This ensures that total atomic balances are preserved. The perturbed initial state $\m{x}_0^{\mathrm{pert}}$ is used to generate the dynamics~\eqref{eq:model_DT}. 
Let $\m{y}_k = \m{C}\m{x}_k$, where $\m{C}$ is the output matrix that maps the concentrations of chemical species $\m{x}_k$ to the measurable states $\m{y}_k$ via the to-be-placed sensors.
\subsection[Solving the partitioning and SP problem on $\mr{(N1)}$]{Solving the partitioning and SP problem on N1}\label{sec:simB}
To demonstrate the proposed framework, we solve the partitioning problem $\mb{P2}$ for $\mr{N1}$ under the $\log\!\det$ observability measure. The choice of this metric is grounded in the results of Theorem~\ref{theo:lowerbound} and its underlying relationship to Kalman filters and, therefore, state-estimation error. The partitioning problem $\mb{P2}$ is then solved using the continuous greedy Algorithm~\ref{alg:continuous_greedy} for $\kappa =4$. We analyze the convergence of the multilinear extension under different sampling numbers and number of steps $T$. A sampling number of $100$ and $T=10$ steps are chosen for solving $\mb{P2}$. We do not show the convergence analysis here for brevity; it depends on the size of the network and the objective function. We note that the combustion networks have non-participating species that are assigned a singleton partition of their own. The SP problem $\mb{P3}$ is subsequently solved for the system under the partitioning $\pi_{\kappa}$ obtained from $\mb{P2}$. The SP results for the partitioned system are compared to those of the full system with no partitions.

The results for solving $\mb{P2}$ and $\mb{P3}$ for the $\mr{H}_2\mr{O}_2$ combustion reaction network are illustrated in Fig.~\ref{fig:H202_results}. The results in Fig.~\ref{fig:H202_unpart} depict the system states when no clustering is considered along with the optimal sensor configuration for $r = 6$ sensors, and in Fig.~\ref{fig:H202_part}, the clustered measurable states along with the sensor configuration for the same number of sensors.
The arrows depict connections between system states obtained from the adjacency matrix representing the list of reactions, $\sum_{i=1}^{n_{x}}q_{ji} \mc{R}_{i} \rightleftarrows  \sum_{i=1}^{n_{x}} w_{ji}\mc{R}_{i}$, of the considered combustion network. Notice that states nine and ten are non-interacting species, i.e., they do not react with any of the other species. To observe such states, a sensor is required to be allocated to each particular state variable. 

By observing Fig.~\ref{fig:H202_results}, we highlight that \textit{(i)} the partitioning, for $\kappa = 4$, results in two non-participating species partitions and two subsystems, $\mc{S}_1$ and $\mc{S}_2$, that maximize partition observability $f(\mc{S}_i)$; and \textit{(ii)} the sensor configuration obtained by solving $\mb{P3}$ using the continuous greedy Algorithm~\ref{alg:continuous_greedy} is the same for both the partitioned and unpartitioned systems. That being said, we obtain the same optimal solution with a $ \logdet{\m{W}(\mc{R})}= \logdet{\sum_{i \in \mc{C}}\m{W}_{\mc{S}_i}(\mc{R})} = 38.91$. As a baseline, the observability measure for a fully sensed $\mr{H}_2\mr{O}_2$ network is $ \logdet{\m{W}} = 79.35$. The computational time for solving the partitioning problem $\mb{P2}$ using Algorithm~\ref{alg:continuous_greedy} is $2.83$ seconds. While solving the partitioning problem using the method of expectation might be costly, a closed form and approximate multilinear extensions exist and can be used~\cite{Iyer2014,Iyer2022}. This reduces the computational complexity by avoiding the random sampling of the expectation.

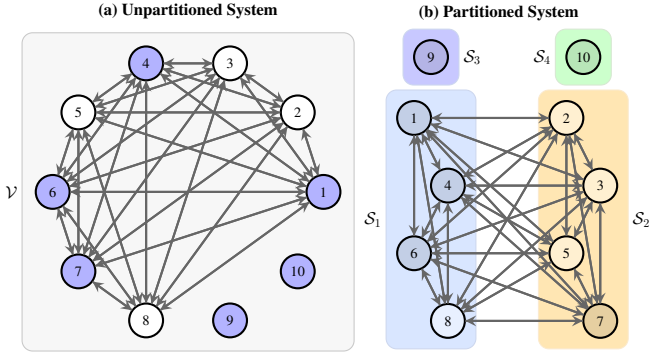
\begin{figure}[t]
	\centering
	\vspace{-0.25cm}
	\captionsetup[subfloat]{labelformat=empty}  
	\makebox[0.2\textwidth]{%
		\subfloat[\label{fig:H202_unpart}]{
		\begin{minipage}{0.2\textwidth}
			\centering
			\begin{tikzpicture}[scale=0.45, transform shape]
				\tikzstyle{node}=[circle, draw=black, thick, fill=white, minimum size=1cm, font=\large]
				\tikzstyle{sensor}=[circle, draw=black, thick, fill=blue!30, minimum size=1cm, font=\large]
				\tikzstyle{arrow}=[->, thick, >=stealth, color=black!60]
				\node[sensor] (x0) at (4.00,0.00) {1};
				\node[node] (x1) at (3.24,2.35) {2};
				\node[node] (x2) at (1.24,3.80) {3};
				\node[sensor] (x3) at (-1.24,3.80) {4};
				\node[node] (x4) at (-3.24,2.35) {5};
				\node[sensor] (x5) at (-4.00,0.00) {6};
				\node[sensor] (x6) at (-3.24,-2.35) {7};
				\node[node] (x7) at (-1.24,-3.80) {8};
				\node[sensor] (x8) at (1.24,-3.80) {9};
				\node[sensor] (x9) at (3.24,-2.35) {10};
				\draw[arrow] (x0) -- (x1);
				\draw[arrow] (x0) -- (x2);
				\draw[arrow] (x0) -- (x3);
				\draw[arrow] (x0) -- (x4);
				\draw[arrow] (x0) -- (x5);
				\draw[arrow] (x0) -- (x6);
				\draw[arrow] (x0) -- (x7);
				\draw[arrow] (x1) -- (x0);
				\draw[arrow] (x1) -- (x2);
				\draw[arrow] (x1) -- (x3);
				\draw[arrow] (x1) -- (x4);
				\draw[arrow] (x1) -- (x5);
				\draw[arrow] (x1) -- (x6);
				\draw[arrow] (x1) -- (x7);
				\draw[arrow] (x2) -- (x0);
				\draw[arrow] (x2) -- (x1);
				\draw[arrow] (x2) -- (x3);
				\draw[arrow] (x2) -- (x4);
				\draw[arrow] (x2) -- (x5);
				\draw[arrow] (x2) -- (x6);
				\draw[arrow] (x2) -- (x7);
				\draw[arrow] (x3) -- (x0);
				\draw[arrow] (x3) -- (x1);
				\draw[arrow] (x3) -- (x2);
				\draw[arrow] (x3) -- (x4);
				\draw[arrow] (x3) -- (x5);
				\draw[arrow] (x3) -- (x6);
				\draw[arrow] (x3) -- (x7);
				\draw[arrow] (x4) -- (x0);
				\draw[arrow] (x4) -- (x1);
				\draw[arrow] (x4) -- (x2);
				\draw[arrow] (x4) -- (x3);
				\draw[arrow] (x4) -- (x5);
				\draw[arrow] (x4) -- (x6);
				\draw[arrow] (x4) -- (x7);
				\draw[arrow] (x5) -- (x0);
				\draw[arrow] (x5) -- (x1);
				\draw[arrow] (x5) -- (x2);
				\draw[arrow] (x5) -- (x3);
				\draw[arrow] (x5) -- (x4);
				\draw[arrow] (x5) -- (x6);
				\draw[arrow] (x5) -- (x7);
				\draw[arrow] (x6) -- (x0);
				\draw[arrow] (x6) -- (x1);
				\draw[arrow] (x6) -- (x2);
				\draw[arrow] (x6) -- (x3);
				\draw[arrow] (x6) -- (x4);
				\draw[arrow] (x6) -- (x5);
				\draw[arrow] (x6) -- (x7);
				\draw[arrow] (x7) -- (x0);
				\draw[arrow] (x7) -- (x1);
				\draw[arrow] (x7) -- (x2);
				\draw[arrow] (x7) -- (x3);
				\draw[arrow] (x7) -- (x4);
				\draw[arrow] (x7) -- (x5);
				\draw[arrow] (x7) -- (x6);
				\begin{scope}[on background layer]
					\node[draw=Vedge, fill=Vgray, 
					rounded corners, inner sep=5pt, fit=(x1) (x2) (x3) (x4) (x5) (x6) (x7) (x8) (x9) (x0), label=left:\Large \(\mc{V}\)] {};
				\end{scope}		
				\node[above=0.5cm] at (0,4.5) {\textbf{\Large (a) Unpartitioned System}};
			\end{tikzpicture}
		\end{minipage}}
		\hspace{0.9cm}
		\subfloat[\label{fig:H202_part}]{
		\begin{minipage}{0.2\textwidth}
			\centering
			\begin{tikzpicture}[scale=0.45, transform shape]
				\tikzstyle{colorednode1}=[circle, draw=black, thick, fill=Sub1color!60, minimum size=1cm, font=\large]
				\tikzstyle{coloredsensor1}=[circle, draw=black, thick, fill=Sub1color!90!black, minimum size=1cm, font=\large]
				\tikzstyle{colorednode2}=[circle, draw=black, thick, fill=Sub2color!60, minimum size=1cm, font=\large]
				\tikzstyle{coloredsensor2}=[circle, draw=black, thick, fill=Sub2color!90!black, minimum size=1cm, font=\large]
				\tikzstyle{colorednode3}=[circle, draw=black, thick, fill=Sub4color!60, minimum size=1cm, font=\large]
				\tikzstyle{coloredsensor3}=[circle, draw=black, thick, fill=Sub4color!90!black, minimum size=1cm, font=\large]
				\tikzstyle{colorednode4}=[circle, draw=black, thick, fill=Sub3color!60, minimum size=1cm, font=\large]
				\tikzstyle{coloredsensor4}=[circle, draw=black, thick, fill=Sub3color!90!black, minimum size=1cm, font=\large]
				\tikzstyle{colorednode5}=[circle, draw=black, thick, fill=Sub3color!60, minimum size=1cm, font=\large]
				\tikzstyle{coloredsensor5}=[circle, draw=black, thick, fill=Sub4color!90!black, minimum size=1cm, font=\large]
				\tikzstyle{arrow}=[->, thick, >=stealth, color=black!60]
				\node[coloredsensor1] (x0) at (-0.50,-0.00) {1};
				\node[coloredsensor1] (x3) at (0.50,-2.00) {4};
				\node[coloredsensor1] (x5) at (-0.50,-4.00) {6};
				\node[colorednode1] (x7) at (0.50,-6.00) {8};
				\node[colorednode2] (x1) at (4.0,-0.00) {2};
				\node[colorednode2] (x2) at (5.0,-2.00) {3};
				\node[colorednode2] (x4) at (4.0,-4.00) {5};
				\node[coloredsensor2] (x6) at (5.0,-6.00) {7};
				\node[coloredsensor3] (x8) at (0.00, 1.8) {9};
				\node[coloredsensor4] (x9) at (4.50,1.8) {10};
				\draw[arrow] (x0) -- (x1);
				\draw[arrow] (x0) -- (x2);
				\draw[arrow] (x0) -- (x3);
				\draw[arrow] (x0) -- (x4);
				\draw[arrow] (x0) -- (x5);
				\draw[arrow] (x0) -- (x6);
				\draw[arrow] (x0) -- (x7);
				\draw[arrow] (x1) -- (x0);
				\draw[arrow] (x1) -- (x2);
				\draw[arrow] (x1) -- (x3);
				\draw[arrow] (x1) -- (x4);
				\draw[arrow] (x1) -- (x5);
				\draw[arrow] (x1) -- (x6);
				\draw[arrow] (x1) -- (x7);
				\draw[arrow] (x2) -- (x0);
				\draw[arrow] (x2) -- (x1);
				\draw[arrow] (x2) -- (x3);
				\draw[arrow] (x2) -- (x4);
				\draw[arrow] (x2) -- (x5);
				\draw[arrow] (x2) -- (x6);
				\draw[arrow] (x2) -- (x7);
				\draw[arrow] (x3) -- (x0);
				\draw[arrow] (x3) -- (x1);
				\draw[arrow] (x3) -- (x2);
				\draw[arrow] (x3) -- (x4);
				\draw[arrow] (x3) -- (x5);
				\draw[arrow] (x3) -- (x6);
				\draw[arrow] (x3) -- (x7);
				\draw[arrow] (x4) -- (x0);
				\draw[arrow] (x4) -- (x1);
				\draw[arrow] (x4) -- (x2);
				\draw[arrow] (x4) -- (x3);
				\draw[arrow] (x4) -- (x5);
				\draw[arrow] (x4) -- (x6);
				\draw[arrow] (x4) -- (x7);
				\draw[arrow] (x5) -- (x0);
				\draw[arrow] (x5) -- (x1);
				\draw[arrow] (x5) -- (x2);
				\draw[arrow] (x5) -- (x3);
				\draw[arrow] (x5) -- (x4);
				\draw[arrow] (x5) -- (x6);
				\draw[arrow] (x5) -- (x7);
				\draw[arrow] (x6) -- (x0);
				\draw[arrow] (x6) -- (x1);
				\draw[arrow] (x6) -- (x2);
				\draw[arrow] (x6) -- (x3);
				\draw[arrow] (x6) -- (x4);
				\draw[arrow] (x6) -- (x5);
				\draw[arrow] (x6) -- (x7);
				\draw[arrow] (x7) -- (x0);
				\draw[arrow] (x7) -- (x1);
				\draw[arrow] (x7) -- (x2);
				\draw[arrow] (x7) -- (x3);
				\draw[arrow] (x7) -- (x4);
				\draw[arrow] (x7) -- (x5);
				\draw[arrow] (x7) -- (x6);
				\node[above=0.5cm] at (2,2.25) {\textbf{\Large (b) Partitioned System}};
				\begin{scope}[on background layer]
					\node[draw=Sub1edge, fill=Sub1color, rounded corners, inner sep=4pt, fit={(x0) (x3) (x5) (x7)}, label=left:\Large $\mc{S}_{1}$] {};
					\node[draw=Sub2edge, fill=Sub2color, rounded corners, inner sep=4pt, fit={(x1) (x2) (x4) (x6)}, label=right:\Large $\mc{S}_{2}$] {};
					\node[draw=Sub4edge, fill=Sub4color, rounded corners, inner sep=4pt, fit={(x8)}, label=right:\Large $\mc{S}_{3}$] {};
					\node[draw=Sub3edge, fill=Sub3color, rounded corners, inner sep=4pt, fit={(x9)}, label=left: \Large $\mc{S}_{4}$] {};
				\end{scope}
			\end{tikzpicture}
		\end{minipage}
	}} \vspace{-0.2cm}
	\caption{$(\mr{N1})$ Partitioning of an $\mr{H}_2\mr{O}_2$ combustion reaction network. Sensor locations on the (a) unpartitioned system and (b) partitioned system that is clustered by color. Sensor nodes can be identified by observing the shaded nodes within each cluster and the colored nodes within the unpartitioned system.}
	\label{fig:H202_results}
\end{figure}

Furthermore, the computational time for solving $\mb{P3}$ using the continuous greedy algorithm and pipage rounding for both the unpartitioned system and partitioned system obtained from solving $\mb{P2}$ are $0.648$ and $0.226$ seconds. This highlights a significant decrease in computational time for the SP problem. The authors in~\cite{Summers2016} show that a simple greedy algorithm can yield $99\%$ optimality for LTI systems. To check the optimality of our solution above, we use a brute-force approach to find the optimal sensor configuration that maximizes the $\log\!\det$ metric. The result is equivalent to the solution obtained from solving $\mb{P3}$ using the proposed framework, thereby validating the theoretical guarantees of Theorem~\ref{theo:guarenteeMulti}. We note that we also obtain the same optimal sensor configuration when solving $\mb{P3}$ using a simple greedy algorithm for the partitioned system under a partition matroid constraint. While the optimality of simple greedy algorithms for partition matroid constraints is $1/2$ in theory, the results herein further demonstrate the practicality of the submodular maximization approach.

Based on the results illustrated on the network $\mr{N1}$ (which serves as an example), we can now address the posed research questions by studying the larger-scale network $\mr{N2}$.
\subsection[Spectral clustering and identifying the range of kappa]{Spectral clustering and identifying the range of $\kappa$}\label{sec:modularity}

\begin{figure}[t]
	\centering
	\hspace*{-0.15cm}
	\includegraphics[width=0.95\columnwidth,keepaspectratio]{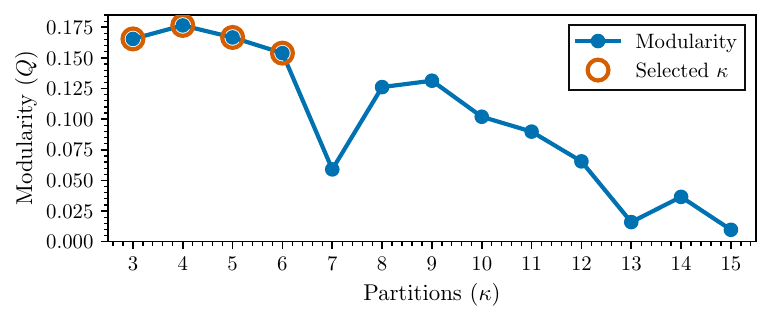}
	\vspace{-0.2cm}
	\caption{Partition modularity $Q$ of the $\mr{GRI}30$ combustion network obtained from spectral clustering for varying $\kappa$. The selected range for $\kappa$ is shown in red.}\label{fig:mod}
\end{figure}

We approach identifying potential values for the partition number $\kappa$ by studying graph-based measures that capture community structure within a network. We do not delve into the literature on such graph-based methods, as this is not the purpose of the work herein. Rather, we borrow a well-known measure that enables the quantification of closely and densely connected subsystems (community structures) as a means $\textit{(i)}$ to motivate the range of values for $\kappa$, and $\textit{(ii)}$ to show that solving the observability-based partitioning problem results in more observable partitions.

We consider a modularity-based metric (not to be confused with Definition~\ref{def:modular_submodular}) to obtain potential choices for the number of partitions. For a system partition $\pi_{\kappa}$, modularity $Q$, which quantifies each subsystem’s interconnection link density~\cite{Newman2004}, can be defined as
\begin{equation}\label{eq:modularityQ}
	Q = \frac{1}{2m} \sum_{i,j=1}^{n_x} \left[ A^{\mr{adj}}_{ij} - \frac{a_i a_j}{2m} \right] \delta(\mathcal{S}_i, \mathcal{S}_j) \in [-1, 1],
\end{equation}
where the adjacency matrix $\m{A}^{\mr{adj}} = [A^{\mr{adj}}_{ij}] \in \mathbb{R}^{n_x \times n_x}$ represents the reaction network $\Theta$, such that $A^{\mr{adj}}_{ij} = 1$ if species $\mathcal{R}_i$ and $\mathcal{R}_j$ participate in any reaction together, and $A^{\mr{adj}}_{ij} = 0$ otherwise. The degree $a_i = \sum_{j=1}^{n_x} A^{\mr{adj}}_{ij}$ represents the number of connections between species $\mathcal{R}_i$ and all other species, while the total number of edges in a network is given by $m = \frac{1}{2}\sum_{i,j=1}^{n_x} A^{\mr{adj}}_{ij}$. The term $\frac{a_i a_j}{2m}$ represents the expected number of connections between species $\mathcal{R}_i$ and $\mathcal{R}_j$ if edges were randomly distributed while preserving node degrees. The function $\delta(\mathcal{S}_i, \mathcal{S}_j) = 1$ if species $\mathcal{R}_i$ and $\mathcal{R}_j$ belong to the same community, and $0$ otherwise.

We consider spectral clustering to obtain partitions for a varying number of $\kappa$ values by computing the normalized Laplacian of the combustion reaction network. The Laplacian matrix for a given adjacency matrix $\m{A}^{\mr{adj}}$ is defined as $\m{L} = \m{D} - \m{A}^{\mr{adj}}\in \mathbb{R}^{n_x \times n_x}$, where $\m{D} \in \mathbb{R}^{n_x \times n_x}$ is a diagonal degree matrix with $D_{ii} = a_i$. The normalized Laplacian is given by $\bar{\m{L}} = \m{D}^{-1/2} \m{L} \m{D}^{-1/2}$. 

\begin{table}[b]
	\fontsize{9}{9}\selectfont
	\centering
	\caption{Difference in optimal solution sets when solving $\mb{P3}$ for the partitioned system with $\kappa = 5$ and $\kappa = 6$, evaluated for $r = 23$ sensors under the global $f(\bigcup_{i \in \mc{C}} \mathcal{R}_i)$ and local $\sum_{i \in \mc{C}} f(\mathcal{R}_i)$ observability objective functions. All other selected states remain the same.}
	\label{tab:partition_diff}
	\vspace{-0.2cm}
	\renewcommand{\arraystretch}{1.3}
	\resizebox{\linewidth}{!}{%
		\begin{tabular}{l|cc|cc}
			\midrule \hline
			\multirow{2}{*}{Partition} & \multicolumn{2}{c|}{$\kappa=5$} & \multicolumn{2}{c}{$\kappa=6$} \\
			& \makecell[c]{$f(\bigcup_{i \in \mc{C}}\mathcal R_i)$ \vspace{0.05cm}\\$ 176.36$}
			& \makecell[c]{$\sum_{i \in \mc{C}} f(\mathcal R_i) \vspace{0.05cm}$\\$ 115.16$}
			& \makecell[c]{$f(\bigcup_{i \in \mc{C}}\mc R_i) \vspace{0.05cm}$\\$ 176.36$}
			& \makecell[c]{$\sum_{i \in \mc{C}} f(\mc R_i) \vspace{0.05cm}$\\$ 106.81$} \\
			\hline
			$\mathcal S_{1}$ & -- & $\{5,\, 20,\, 30\}$ & -- & -- \\
			$\mathcal S_{2}$ & $\{18,\, 29,\, 33,\, 45\}$ & -- & -- & $\{20,\, 30\}$ \\
			$\mathcal S_{3}$ & -- & $\{27\}$ & $\{33\}$ & $\{53\}$ \\
			$\mathcal S_{4}$ & -- & -- & $\{23,\, 29\}$ & -- \\
			$\mathcal S_{5}$ & -- & -- & -- & -- \\
			$\mathcal S_{6}$ & $\emptyset$ & $\emptyset$ & -- & -- \\
			\toprule \bottomrule
		\end{tabular}
	}
\end{table}
\setlength{\textfloatsep}{1pt}

For each $\kappa$, we compute the first $\kappa$ eigenvectors corresponding to the smallest non-zero eigenvalues of $\bar{\m{L}}$ and use them to cluster the species into $\kappa$ partitions using the spectral $k$-means clustering algorithm~\cite{Ng2001}. We vary $\kappa$ from $3$ (due to the existence of non-interacting species) to $14$ and compute the corresponding modularity $Q$ scores for each partitioning $\pi_{\kappa}$. The modularity scores for $\mr{N2}$ are shown in Fig.~\ref{fig:mod}. We observe that $Q$ values for $\kappa \in \{3, 4, 5, 6\}$ correspond to high modularity scores ($Q \in [0.143, 0.175]$), with $\kappa = 4$ achieving the maximum modularity of $Q = 0.175$. Note that high modularity ($Q > 0$) indicates that nodes within the same partition are more densely connected, revealing natural community structure. Based on the results, to solve $\mb{P2}$ and evaluate its optimality, we select $\kappa \in \{3, 4, 5, 6\}$.

\begin{figure*}[ht]
	\centering
	\captionsetup[subfloat]{labelformat=empty}
	\subfloat[\label{fig:unpartitioned_GRI30}]{\includegraphics[width=0.33\textwidth]{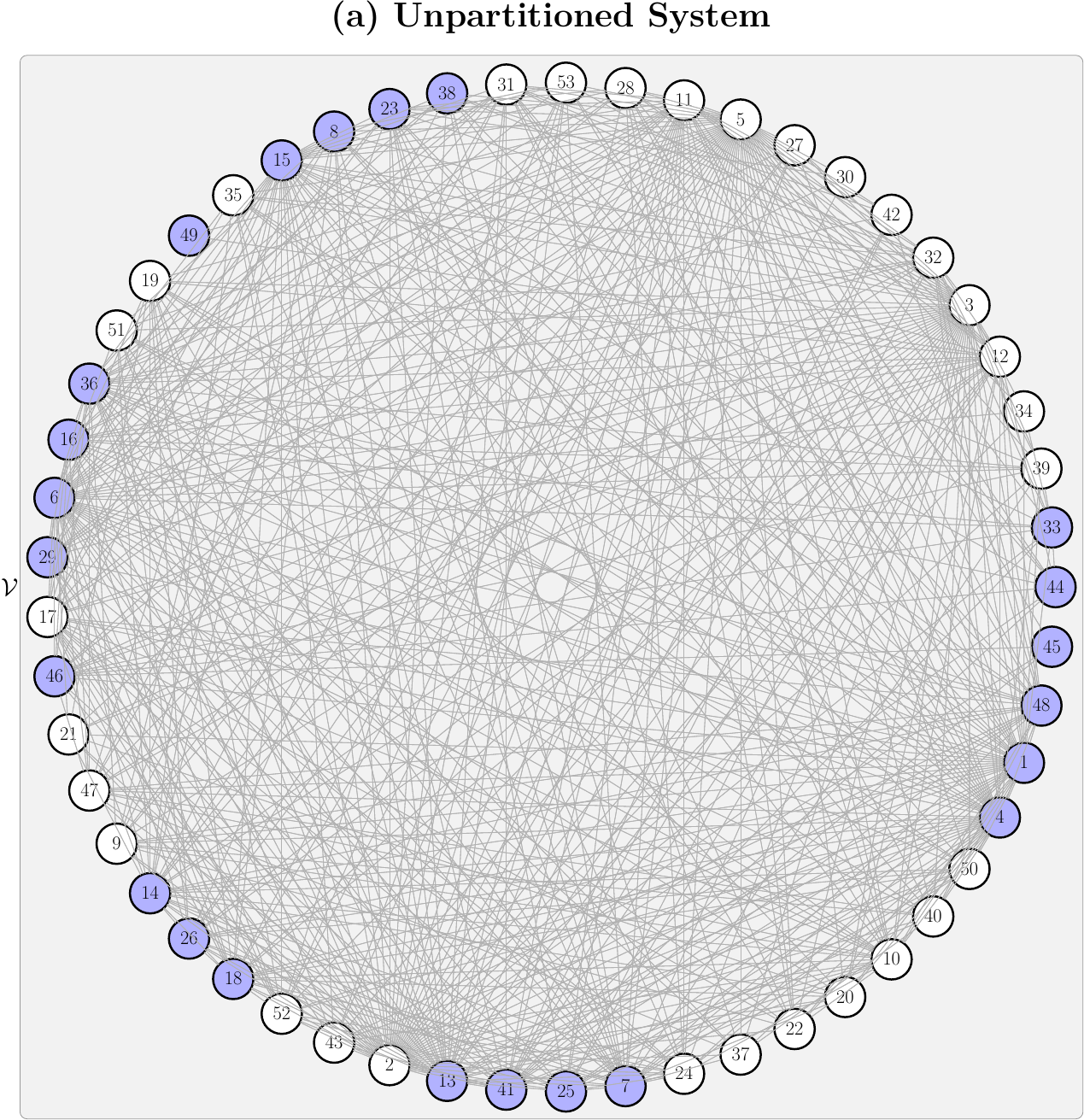}}
	\hspace*{0.1cm}
	\subfloat[\label{fig:partitioned_k2}]{\includegraphics[width=0.478\textwidth]{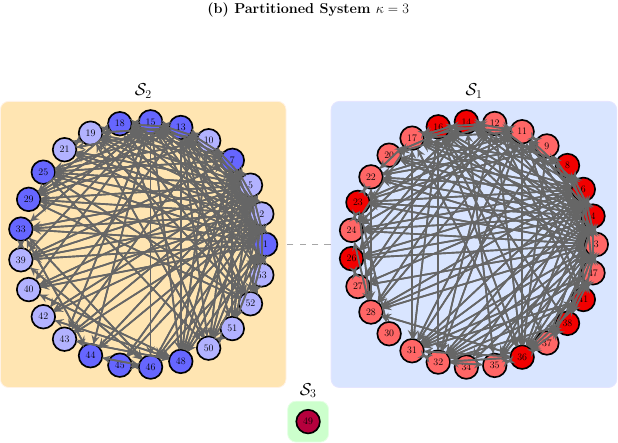}}
	\vspace*{-0.75cm}
	\subfloat[\label{fig:partitioned_k3}]{\includegraphics[width=0.33\textwidth]{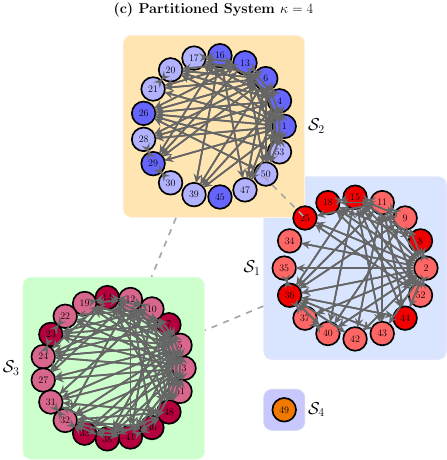}}{}{} 
	\subfloat[\label{fig:partitioned_k4}]{\includegraphics[width=0.33\textwidth]{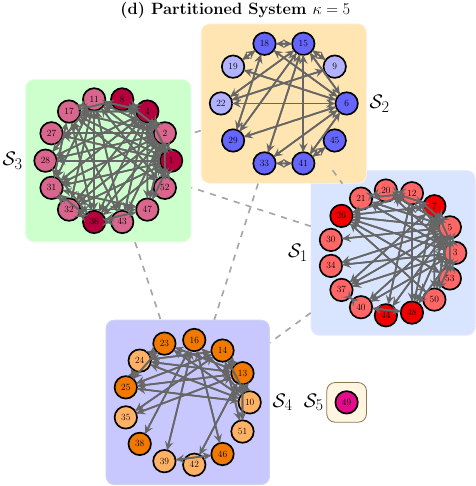}}
	\subfloat[\label{fig:partitioned_k6}]{\includegraphics[width=0.33\textwidth]{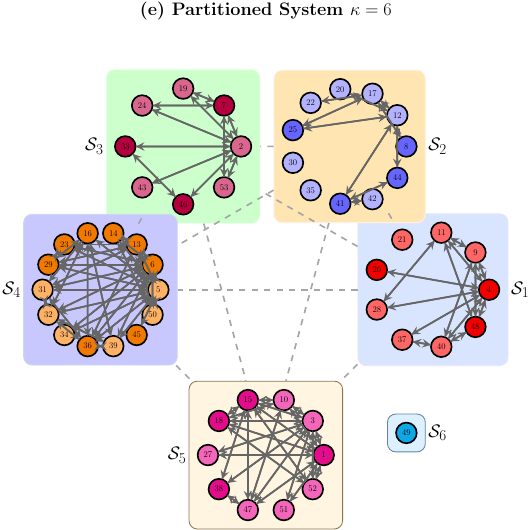}}
	\vspace{-0.2cm}
	\caption{$(\mr{N2})$ Partitioning of a $\mr{GRI30}$ combustion reaction network. Sensor locations for the (a) unpartitioned system and partitioned systems with (b) $\kappa=3$, (c) $\kappa=4$, (d) $\kappa=5$, and (e) $\kappa=6$ are shown. Sensor nodes can be identified by observing the shaded nodes within each subsystem and the colored nodes within the unpartitioned system. The state network within the subsystem is shown by the arrows. The dashed arrows illustrate subsystems that have at least one state interacting with a state in another partition. Note that partitions (b) $\mc{S}_3$, (c) $\mc{S}_4$, (d) $\mc{S}_5$ and (e) $\mc{S}_6$ are singleton partitions indicating non-participating species.}
	\label{fig:GRI30_results}
	\vspace{-0.6cm}
\end{figure*}

\begin{figure}[t]
	\centering
	\includegraphics[keepaspectratio=true,width=\columnwidth]{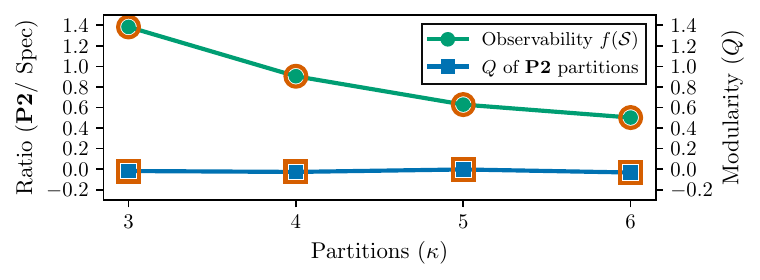}
	\vspace{-0.5cm}
	\caption{Modularity $Q$ of the optimal partitions and the ratio of the observability measure $f(\mathcal{S})$ for subsystems obtained by solving $\mb{P2}$ versus those obtained via spectral clustering for the $\mr{GRI}30$ network.} \label{fig:mod-obs}
\end{figure}

\subsection{Answering the posed research questions}
To answer the posed research questions, we now solve $\mb{P2}$ and $\mb{P3}$ for the larger $\mr{GRI30}$ combustion reaction network. The results for partitioning under partition numbers $\kappa = \{3,4,5,6\}$ and $r = 23$ are depicted in Fig.~\ref{fig:GRI30_results}. The computational time for solving the partitioning problem under each $\kappa$ is $\{11.997, 21.097, 34.415, 49.648\}$ seconds; see Section~\ref{sec:multilinear} for the time complexity. Notice that by increasing $\kappa$, the partitions obtained tend to have nodes/states that are not interacting (reacting) with the states within the same partition; see Fig.~\ref{fig:partitioned_k6}, partition $\mc{S}_2$. This is expected since from Fig.~\ref{fig:mod}, we can see that modularity for $\kappa>5$ is low. The reason is that the partition number $\kappa = 6$ is larger than that of the communities of species that interact within the combustion reaction network. When increasing $\kappa$ beyond the existing community structures, states are acquired by the partition while increasing the overall system observability abilities of the partition. This is corroborated by the fact that states $\{30, 35\} \in \mc{S}_2$ are interacting with sensed species in other partitions. Observing Fig.~\ref{fig:partitioned_k4}, partition $\mc{S}_2$, we see that sensor state $\{29\}$ now in $\mc{S}_4$ (Fig.~\ref{fig:partitioned_k6}) is interacting with state $\{30\}$, while sensed state $\{26\}\in \mc{S}_1$ is interacting with state $\{35\}$. Thus by construction, $\mb{P2}$ partitions the system while maximizing the overall system observability by maintaining inter-subsystem connections. In other words, we obtain species within a subsystem that do not interact with the species in the same subsystem but interact with other subsystems, which means the partitioning problem favors overall observability maximization when $\kappa$ is not optimal for inter-subsystem observability.

We note that as a result of the practical optimality of the $\mb{P3}$ for $\mr{N1}$, we solve $\mb{P3}$ for $\mr{N2}$ using the simple greedy algorithm. Similar to the results obtained for the small combustion reaction network, the optimal sensor configuration for the unpartitioned system is equivalent to that of the partitioned system. Furthermore, notice that we obtain the same optimal configuration under different partition numbers. This is a result of solving $\mb{P3}$ under objective function defined as $\logdet{\sum_{i \in \mc{C}} \m{W}_{\mc{S}_i}(\mc{R}_i)}$. Now, to validate Theorem~\ref{theo:lowerbound}, we solve the SP problem for objective function $\sum_{i \in \mc{C}}\logdet{\m{W}_{\mc{S}_i}(\mc{R}_i)}$. The difference in sensor configuration results with respect to those obtained from $\logdet{\sum_{i \in \mc{C}} \m{W}_{\mc{S}_i}(\mc{R}_i)}$ are summarized in Tab.~\ref{tab:partition_diff}.

Notice that the resulting sensor configurations for both $\kappa$ values are suboptimal compared to the original objective function; however, the configuration for $\kappa = 5$ is closer to the optimal value than that for $\kappa = 6$. By observing the difference in chosen states within specific partitions, the choice of states favors localized observability, meaning that the chosen states have more intra-subsystem connections (species interacting within the same cluster). This can be easily depicted for $\kappa = 5$ by checking state $\{5\}$ in $\mc{S}_2$ (Fig.~\ref{fig:partitioned_k4}), which is chosen as a sensor node when solving for the local objective function. This state interacts with several other states with the same partition, thereby enhancing the observability of that subsystem. This can also be observed for $\kappa = 6$ by checking states $\{20, 30\}$ in $\mc{S}_2$ and $\{53\}$ in $\mc{S}_3$ (Fig.~\ref{fig:partitioned_k6}). However, since the subsystem Gramian remains of size $\Rn{n_x \times n_x}$, the localized objective function still contributes to maximizing global observability. The partitioning of the Gramian into submatrices of different sizes results in subsystems that are disconnected (i.e., with no interconnections). This is a topic for future research, as it requires reformulating the partitioning problem $\mb{P2}$.

\begin{table*}[t]
	\vspace*{0.15cm}	
	\fontsize{9}{9}\selectfont
	\centering
	\caption{Optimality and computational time for solving the SP problem $\mb{P3}$ over the partitioned and unpartitioned $\mr{GRI}30$ combustion reaction network for varying $\kappa$, $r$, and objective function formulation.}
	\label{tab:comp_time}
	\renewcommand{\arraystretch}{1.3}
	\resizebox{\linewidth}{!}{%
		\begin{tabular}{l|ccc|ccc|ccc|ccc}
			\midrule \hline
			\multirow{2}{*}{$\mr{N2}$: $\mr{GRI}30$} & \multicolumn{3}{c|}{$\kappa=3$} & \multicolumn{3}{c|}{$\kappa=4$} & \multicolumn{3}{c|}{$\kappa=5$} & \multicolumn{3}{c}{$\kappa=6$} \\
			\cline{2-13}
			& $r=10$ & $r=16$ & $r=23$ & $r=10$ & $r=16$ & $r=23$ & $r=10$ & $r=16$ & $r=23$ & $r=10$ & $r=16$ & $r=23$ \\ \hline
			
			$f(\mathcal R)$ & -50.12 & 91.00 & 176.36 & -50.12 & 91.00 & 176.36 & -50.12 & 91.00 & 176.36 & -50.12 & 91.00 & 176.36 \\
			$\mathrm{time}\;(\mathrm{sec})$ & 12.10 & 18.49 & 24.35 & 11.95 & 17.98 & 24.48 & 12.12 & 18.63 & 24.19 & 12.05 & 18.62 & 24.20 \\
			\hline
			
			$f(\cup_{i\in\mathcal C}\mathcal R_i)$ & -50.12 & 91.00 & 176.36 & -50.12 & 91.00 & 176.36 & -50.12 & 91.00 & 176.36 & -50.12 & 91.00 & 176.36 \\
			$\mathrm{time}\;(\mathrm{sec})$ & 1.89 & 2.84 & 3.67 & 1.80 & 2.86 & 3.70 & 1.92 & 2.88 & 3.58 & 1.94 & 2.83 & 3.55 \\
			\hline
			
			$\sum_{i\in\mathcal C} f(\mathcal R_i)$ & -93.27 & 49.69 & 136.85 & -115.22 & 18.97 & 123.43 & -128.88 & 13.41 & 115.16 & -136.02 & 1.43 & 106.81 \\
			$\mathrm{time}\;(\mathrm{sec})$ & 1.95 & 2.91 & 3.80 & 1.82 & 2.97 & 3.83 & 1.97 & 3.04 & 3.87 & 1.99 & 2.93 & 4.30 \\
			\toprule \bottomrule
	\end{tabular}}
	\setlength{\textfloatsep}{1pt}
	\vspace*{-0.2cm}
\end{table*}

Based on the results, we compare the observability measure and the modularity $Q$ of the resulting partitions in Fig.~\ref{fig:GRI30_results} with those obtained using spectral clustering in Section~\ref{sec:modularity}. Observe from Fig.~\ref{fig:mod-obs} that the observability ratio $f_{\mb{P2}}(\mc{S}) / f_{\text{spec}}(\mc{S})$ of partitions obtained from solving $\mb{P2}$ versus those from spectral clustering are $1.4$ to $0.6$ times higher for the considered $\kappa$ values. However, the respective modularity values of the optimal partitions decrease to around zero for each $\kappa$; this is in contrast to those of the spectral clustering partitions. This shows that $\mb{P2}$ results in more observable partitions while balancing inter/intra-connections within the subsystems ($Q \approx 0$), thus demonstrating the value of observability-based partitioning in identifying subsystems that are more observable and can potentially enhance state estimation performance. The results also provide evidence that higher modularity (i.e., more densely connected subsystems) does not necessarily imply improved observability of the subsystems.

Furthermore, Tab.~\ref{tab:comp_time} shows the optimality of the objective function and the computational time for solving $\mb{P3}$ for different partition numbers $\kappa$ and sensor numbers $r$. Accordingly, we highlight the following results. The optimal value for $\mb{P3}$ for the unpartitioned system, $f(\mc{R})$, and the partitioned system, $f(\cup_{i \in \mc{C}} \mc{R}_i)$, are equivalent. The computational time is significantly reduced for all the cases studied under system partitioning; however, it does not change with varying the value of $\kappa$ but with the number of sensors. This is due to a computational threshold, from $\log\!\det$ evaluations, that is not exceeded for this system size. The optimal value for $\mb{P3}$ under partitioning and the localized objective function $\sum_{i \in \mc{C}} f(\mathcal{R}_i)$ is suboptimal for all cases and thereby validates Theorem~\ref{theo:lowerbound}. The least suboptimal values obtained for all the sensor numbers studied are for $\kappa = 3$. This shows that, for this network, the partitioning problem $\mb{P2}$ favors maximizing overall system observability by considering optimally allocating states to a partition with more inter-subsystem connections.

To better understand and further validate the proposed partitioning and SP framework, we assess the performance of the optimal sensor configurations from a state-estimation perspective. We implement a discrete-time Kalman filter~\cite{Kalman1960} for observation horizon $N = 1000$ with Monte Carlo averaging over $50$ trials to assess statistical performance. The noise covariance matrices are chosen as $Q = 10^{-4} \eye_{n_x}$ (process noise) and $R = 10^{-4} \eye_{n_y}$ (measurement noise). The initial error covariance is set to $P_0 = 10^{-1} I_{n_x}$. The state estimation relative error for varying partition numbers $\kappa$ and sensors $r$ are shown in~Fig.~\ref{fig:state-estimation-results}. The relative error is computed as $\|\hat{\m{x}} - \m{x}\|_2 / \|\m{x}\|_2$, where $\hat{\m{x}}$ is the estimated state and $\m{x}$ is the true state. We observe that the state-estimation performance for the partitioned system under different $\kappa$ values matches that of the unpartitioned system. This further validates the optimality of the sensor configurations obtained by solving $\mb{P3}$ under the proposed framework. Notice that increasing the number of sensors $r$, which is equivalent to achieving a higher observability measure (Tab.~\ref{tab:comp_time}), leads to improved state-estimation performance. This is expected due to the underlying relationship between $\log\!\det$ and state-estimation error. While the results show that partitioning does not degrade estimation performance, they warrant future investigation into the performance of localized or distributed state estimators, particularly for boundary states.

 \begin{figure}[t]
 	\centering
 	\hspace*{-0.15cm}
 	\includegraphics[keepaspectratio=true,width=\columnwidth]{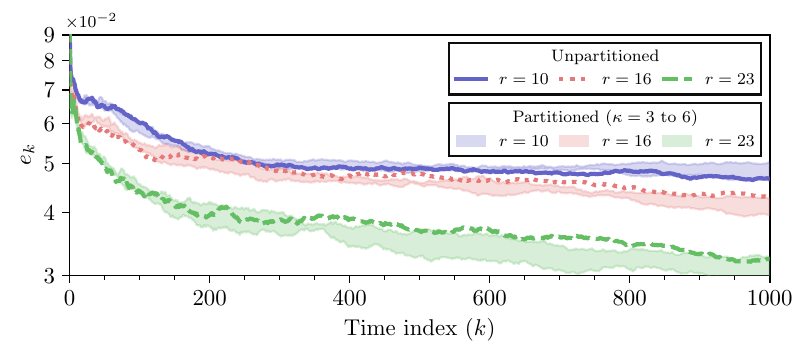}
 	\vspace{-0.5cm}
 	\caption{State estimation results over the partitioned and unpartitioned $\mr{GRI}30$ network under optimal SP partitioned for varying $\kappa$, and $r$.}
 	\label{fig:state-estimation-results}
 \end{figure}

Based on the numerical results of two combustion reaction networks, we conclude the following. \textit{(i)} That the results for the SP problem over partitioned and unpartitioned systems yield similar optimal values while significantly reducing the computational time; this answers $\mr{Q1}$. \textit{(ii)} The above results validate the theoretical bounds established for the $\log\!\det$ metric in Theorem~\ref{theo:lowerbound}. We have shown that solving $\mb{P3}$ under the localized objective function $\sum_{i \in \mc{C}} \log\!\det(\m{W}_{\mc{S}_i}(\mc{R}_i))$ yields suboptimal values when compared to the global objective $f(\cup_{i \in \mc{C}} \mc{R}_i)$; this answers $\mr{Q2}$. \textit{(iii)} By varying the number of partitions $\kappa$, the structure of the partitions and observability can be affected. For network $\mr{N2}$, the optimal partition number is $\kappa = 5$. While achieving the highest observability-based optimal value $f(\mc{S})$, it also demonstrates better performance under parameterized local observability measures. This shows that the partitioning and SP framework supports scalable partitioning while offering trade-offs between subsystem intraconnections and the number of partitions; this answers $\mr{Q3}$. \textit{(iv)} By comparing to spectral clustering-based partitioning, we have shown that the proposed observability-based partitioning framework yields partitions that are more observable while balancing inter/intra-subsystem connections; this is also reflected in the state-estimation performance. This answers $\mr{Q4}$ and concludes this section.
\section{Final Remarks}\label{sec:conclusion}
This paper presented a new perspective for solving the SP problem. We introduced a partitioning and sensor placement framework grounded in submodular maximization formulations. By clustering the measurable state-space into a user-defined number of interconnected partitions, we enable scalable solutions to the SP problem under a partition matroid while maximizing full system observability. The global and local subsystem observability under sensor parameterization are studied, thereby offering insights into the trade-offs introduced by partitioning. By duality, all of the results hold for actuator placement under partitioned system dynamics using controllability Gramian-based measures. This work is not devoid of limitations. Specifically, we do not consider: $\textit{(i)}$ process noise or measurement noise in the problem formulation, although we study state-estimation performance under such configurations; $\textit{(ii)}$ redundancy towards sensor failure events, which impacts practical deployment; and $\textit{(iii)}$ models that depict infrastructure system networks. The work presented thus merits future investigation on this topic, which entails: \textit{(a)} addressing the aforementioned limitations; \textit{(b)} studying the partitioning problem while explicitly removing inter-partition connections, thereby decomposing the Gramian into sub-Gramians; and \textit{(c)} solving $\mb{P3}$ utilizing algorithms that further exploit the partition structure (distributed algorithms~\cite{Rezazadeh2021}).

\bibliographystyle{IEEEtran}
\bibliography{library}
\balance
\end{document}